%% file: main.tex
\documentclass[journal,twoside,web]{ieeecolor}
\input{3-compilation/packages}
\input{3-compilation/user_specific}

\begin{document}
\input{2-info_paper/titlepage}

\maketitle
\input{1-major_content/0-abstract}
\input{1-major_content/1-introduction}
\input{1-major_content/2-preliminaries}
\input{1-major_content/3-problem_formulation}

\input{1-major_content/4-theoretical_analysis}
\input{1-major_content/4.5-main_results}

\input{1-major_content/5-case_study}

\input{1-major_content/6-conclusion}
\input{1-major_content/7-appendices}

\vspace*{-0.5cm}
\section*{References}
\vspace*{-0.5cm}
{\footnotesize
\bibliographystyle{IEEEtran}
\bibliography{4-references/additional_reference,4-references/basic_references,4-references/books,4-references/cbf,4-references/event_mpc,4-references/learning,4-references/microgrid,4-references/network,4-references/oco,4-references/intro_used,4-references/ourpaper}
}

\end{document}

%% file: 3-compilation/packages.tex
\usepackage{generic}
\usepackage{epsfig} 

\usepackage{cite}
\usepackage{amsmath,amssymb,amsfonts}
\usepackage{mathtools}

\makeatletter
\let\proof\@undefined                        
\let\endproof\@undefined                  
\makeatother
\usepackage{amsthm}

\usepackage{algorithmic}
\usepackage{graphicx}
\usepackage{transparent} 
\usepackage{algorithm,algorithmic}
\usepackage{dsfont} 
\usepackage{bm} 
\usepackage{makecell} 
\usepackage{tikz} 
\usepackage{subcaption}
\usepackage[dvipsnames]{xcolor}
\let\labelindent\relax
\usepackage{enumitem}

\allowdisplaybreaks

\usepackage{hyperref}
\hypersetup{
    colorlinks,
    linkcolor={blue},
    citecolor={blue},
    urlcolor={blue}
} 
\usepackage{textcomp}
\def\BibTeX{{\rm B\kern-.05em{\sc i\kern-.025em b}\kern-.08em
    T\kern-.1667em\lower.7ex\hbox{E}\kern-.125emX}}

%% file: 3-compilation/user_specific.tex
\theoremstyle{plain} 
\newtheorem{definition}{\textbf{Definition}}
\newtheorem{assumption}{\textbf{Assumption}}
\newtheorem{theorem}{\textbf{Theorem}}
\newtheorem{lemma}{\textbf{Lemma}}
\newtheorem{proposition}{\textbf{Proposition}}

\theoremstyle{definition} 
\newtheorem{remark}{\textbf{Remark}}

\newcommand{\bR}{\mathbb{R}}

\newcommand{\bN}{\mathbb{N}}
\newcommand{\bI}{\mathbb{I}}

\newcommand{\cX}{\mathcal{X}}
\newcommand{\cU}{\mathcal{U}}

\newcommand{\cT}{\mathcal{T}}
\newcommand{\cJ}{\mathcal{J}}

\newcommand{\cK}{\mathcal{K}}
\newcommand{\cKL}{\mathcal{K}\mathcal{L}}

\newcommand{\Df}{\Delta f}
\newcommand{\mmf}{\varepsilon_f}
\newcommand{\els}{\ell^\star}
\newcommand{\elf}{\ell_{\mathrm{T}}}
\newcommand{\epf}{\kappa_{\mathrm{T}}}

\newcommand{\eigmax}[1]{\overline{\lambda}_{#1}}
\newcommand{\eigmin}[1]{\underline{\lambda}_{#1}}

\newcommand{\splus}{\hspace{-0.08cm}+\hspace{-0.08cm}}
\newcommand{\sto}{\hspace{-0.08cm}\to\hspace{-0.08cm}}
\newcommand{\sminus}{\hspace{-0.06cm}-\hspace{-0.06cm}}
\newcommand{\seq}{\hspace{-0.08cm}=\hspace{-0.08cm}}

\newcommand{\sleq}{\hspace{-0.08cm}\leq\hspace{-0.08cm}}

\newcommand{\bzero}{\mathbf{0}}
\newcommand{\bone}{\mathbf{1}}
\newcommand{\bfI}{\mathbf{I}}

\newcommand{\Rcrnom}{R_{N,M}}
\newcommand{\Rcr}{R^\ast_{N,M}(\mmf)}

\newcommand{\pit}{\pi_{\mathrm{T}}}

\newcommand{\tf}{f^\ast}
\newcommand{\tsys}[1]{f^\ast_{#1}}
\newcommand{\tsysc}{F^\ast_{0,\infty}}
\newcommand{\tsysn}{F^\ast_{t,N}}

\newcommand{\cost}{\ell}

\newcommand{\probopt}[2]{\mathrm{P}_{\mathrm{IHOPC}}(#1;{#2})}
\newcommand{\probmpc}[2]{\mathrm{P}_{\mathrm{MPC}}(#1;{#2})}

\newcommand{\umpc}{\overline{\gamma}}
\newcommand{\mpc}{\gamma}

\newcommand{\bu}{\mathbf{u}}
\newcommand{\bv}{\mathbf{v}}

\newcommand{\xroa}{\mathcal{X}_{\mathrm{ROA}}}

\newcommand{\nQ}[1]{\| #1 \|^2_Q}

\newcommand{\nP}[1]{\| #1 \|^2_P}



%% file: 2-info_paper/titlepage.tex
\title{Asynchronous Model Predictive Control Under Model Mismatch: Stability and Performance Guarantees}
\author{Changrui Liu, \IEEEmembership{Graduate Student Member, IEEE}, Anil Alan, Shengling Shi, \IEEEmembership{Member, IEEE}, \\and Bart De Schutter, \IEEEmembership{Fellow, IEEE}
\thanks{This paper is part of a project that has received funding from the European Research Council (ERC) under the European Union’s Horizon 
2020 research and innovation programme (Grant agreement No. 101018826 - CLariNet). (Corresponding authors: Changrui Liu \& Shengling Shi)}
\thanks{Changrui Liu, Anil Alan, Shengling Shi, and Bart De Schutter are with the Delft Center for Systems and Control, TU Delft, Delft, The Netherlands. (e-mail: \{c.liu-14, a.alan, shengling.shi, b.deschutter\}@tudelft.nl). }
}

%% file: 1-major_content/0-abstract.tex
\begin{abstract}
Certainty-equivalence model predictive control (CE-MPC) is widely used for its simplicity and efficiency, but theoretical guarantees under asynchronous feedback remain limited. This paper establishes stability and performance guarantees for asynchronous CE-MPC of input-constrained nonlinear systems. We first derive a nominal stability condition and competitive-ratio bound that explicitly account for inter-execution intervals without prescribing a feedback mechanism. A value-function perturbation analysis for quadratic stage costs then accommodates additive, potentially non-smooth model mismatch without constraint qualification conditions. Combining these results yields stability criteria and competitive-ratio bounds for CE-MPC under general asynchronous feedback, including event/self-triggered and multi-step MPC. The guarantees explicitly relate prediction horizon, inter-execution time, and uncertainty magnitude, quantifying performance degradation relative to an ideal infinite-horizon controller. These results clarify tradeoffs between feedback frequency, model accuracy, and horizon length, guiding asynchronous MPC design using approximate or learned models.
\end{abstract}

\begin{IEEEkeywords}
Nonlinear model predictive control, event-based control, performance analysis, uncertain systems 
\end{IEEEkeywords}

%% file: 1-major_content/1-introduction.tex
\section{Introduction}
\label{sec:introduction}
\IEEEPARstart{M}{odel} predictive control (MPC) has been successfully deployed across diverse engineering domains, including robotics~\cite{worthmann2015model}, power electronics~\cite{wang2020event}, and traffic control~\cite{wu2020distributed}. Traditional implementations rely on high-frequency, time-triggered feedback, which can be computationally costly and impractical for resource-constrained systems~\cite{li2014event, brunner2017robust}. This limitation motivates asynchronous feedback designs, such as event/self-triggered MPC~\cite{sun2019robust} and multi-step MPC~\cite{grune2010analysis}, which reduce computational and communication overhead by recomputing control inputs only at selected time instants~\cite{eqtami2010event, heemels2012introduction}. Crucially, MPC relies on an internal prediction model to compute these optimal inputs~\cite[Section 1.2]{rawlings2017model}. In practice, however, obtaining a perfect model is impossible due to unmodeled dynamics, limited data, or measurement noise. Consequently, plant-model mismatch is unavoidable. The irregular timing of asynchronous feedback significantly exacerbates the challenges of handling this mismatch, creating a coupled effect that complicates stability and suboptimality analysis, which motivates the theoretical framework developed in this work.

To address plant–model mismatch, prior work has extensively studied MPC combined with online model identification~\cite{mesbah2022fusion}, focusing on stability and feasibility guarantees~\cite{kohler2020computationally, kohler2021robust}. In contrast, relatively few studies provide closed-loop suboptimality guarantees~\cite{wabersich2022cautious, liu2025regret, degner2026adaptive}. Even when online identification is incorporated, residual model errors persist, making it natural to analyze MPC within the certainty-equivalence framework~\cite{soloperto2019dual}, where the MPC controller relies on a nominal model obtained offline. Recent advances in analyzing the suboptimality of CE-MPC relative to infinite-horizon optimal control (IHOPC) with perfect model knowledge have focused on linear systems~\cite{liu2024stability, shi2025suboptimality} and nonlinear systems~\cite{schwenkel2020robust, liu2026certainty}. However, most existing analyses assume time-triggered feedback and do not explicitly account for asynchronous feedback inherent in event/self-triggered and multi-step implementations.

Event/self-triggered MPC provides a principled method to handle model mismatch~\cite{brunner2017robust}. By bounding the deviation between predicted and true states, state-dependent triggering rules can be designed to recompute control inputs only when this deviation exceeds a prescribed threshold~\cite{sun2019robust, sun2021dynamic}. In addition to stability and feasibility results, performance bounds based on relaxed dynamic programming~\cite{lincoln2006relaxing} have been established to quantify the infinite-horizon performance of asynchronous MPC~\cite{lu2015asynchronous, lu2022self}. However, existing approaches typically rely on intricate triggering rules involving tuning parameters~\cite{sun2019robust, lu2015asynchronous, lu2022self}, and the resulting suboptimality bounds are expressed in terms of those parameters, thereby obscuring the direct influence of the inter-execution steps (i.e., feedback frequency) on the closed-loop performance of MPC. In this work, we consider a more general asynchronous feedback framework, making our analysis independent of the specific triggering mechanisms. Moreover, existing analyses for event/self-triggered MPC often require the terminal cost to be a control Lyapunov function (CLF)~\cite{sun2019robust, sun2021dynamic}. While standard in MPC theory~\cite{rawlings2012fundamentals}, this requirement can be restrictive because constructing such terminal functions is challenging~\cite{sun2019robust, kohler2023stability}, particularly under model mismatch. Despite the clear benefits of asynchronous feedback, a unified framework characterizing the stability and suboptimality of asynchronous CE-MPC under model mismatch, with explicit dependence on the inter-execution steps, remains absent from the literature.

In line with recent results on the inherent robustness of MPC~\cite{pannocchia2011conditions} and MPC with general terminal costs~\cite{kohler2023stability}, this work examines the stability and infinite-horizon performance of general asynchronous CE-MPC for input-constrained nonlinear systems. The results of this paper are closely related to the stability and performance analysis of multi-step MPC under perfect model assumptions~\cite{grune2010analysis},~\cite[Chapter 10.4]{grune2017nonlinear}; however, we explicitly consider model mismatch in this work. More specifically, we address the following foundational questions:
\begin{enumerate}
\item \textit{Given a perfect model, how do the prediction horizon, and the inter-execution steps jointly influence closed-loop stability and infinite-horizon performance?}
\item \textit{Given plant-model mismatch, under what conditions does asynchronous CE-MPC stabilize the true system?}
\item \textit{Given plant-model mismatch, how do model mismatch, the prediction horizon, and inter-execution steps jointly affect the infinite-horizon performance?}
\end{enumerate}

In response, this paper advances the state of the art in multi-step MPC~\cite{grune2010analysis, grune2017nonlinear} and time-triggered CE-MPC~\cite{liu2026certainty} through the following three contributions:
\begin{enumerate}
\item We analyze asynchronous MPC under nominal conditions, explicitly quantifying how the prediction horizon and maximum inter-execution step govern the stability and the infinite-horizon performance without restricting the analysis to specific asynchronous feedback mechanisms.
\item We establish a stability condition for asynchronous CE-MPC that explicitly accounts for the prediction horizon, the maximum inter-execution interval, and the model mismatch, thereby quantifying how the mismatch and asynchronous feedback affect the stability.
\item We develop a novel pipeline to establish the infinite-horizon performance of asynchronous CE-MPC via perturbation analysis. Specifically, we derive a competitive-ratio bound, expressed conceptually as:
\begin{equation*}
R \geq \frac{\text{Cost of asynchronous CE-MPC}}{\text{Cost of IHOPC}}.
\end{equation*}
This bound quantifies the suboptimality gap relative to IHOPC as a function of the prediction horizon, the maximum inter-execution step, and the model mismatch.
\end{enumerate}

Importantly, our suboptimality bounds are consistent with classical results obtained in the absence of model mismatch and asynchronous execution~\cite{grune2017nonlinear, kohler2023stability}, as they confirm that using more frequent feedback (i.e., shorter inter-execution steps) and longer prediction horizons reduces the suboptimality gap. We emphasize that this work does not propose new MPC algorithms; rather, it provides a unifying theoretical stability and suboptimality analysis of asynchronous CE-MPC for uncertain nonlinear systems under the condition that the model mismatch does not change the equilibrium of the true model~\cite{schimperna2025data, liu2026certainty}. 

The remainder of the paper is organized as follows. Section~\ref{sec:2-preliminaries} introduces preliminaries and notation. Section~\ref{sec:3-problem_formulation} formulates the problem. Sections~\ref{sec:4-theoretical_analysis} and~\ref{sec:theoreticalAnalysis_model_mismatch} present the stability and performance analysis. Section~\ref{sec:5-examples} reports numerical simulations, and Section~\ref{sec:6-conclusion} concludes the paper.

%% file: 1-major_content/2-preliminaries.tex
\vspace{-0.2cm}
\section{Preliminaries}
\label{sec:2-preliminaries}
\input{1.2-preliminaries/1.2.1-notation}
\input{1.2-preliminaries/1.2.2-basic}

\input{1.2-preliminaries/1.2.3-st_mpc}

%% file: 1.2-preliminaries/1.2.1-notation.tex
\subsection{Notation}
\label{subsec2.1-notation}
The set of (non-negative) real numbers is denoted by $\bR$ ($\bR_+$). The set of natural numbers is $\bN$, and $\bI_{a:b} := \bN \cap [a, b]$ for $0 \leq a \leq b$. 
The symbols $\bzero_n$, $\bone_n$, and $\bfI_n$ denote the zero vector, one vector, and identity matrix of dimension $n$. The symbol $\|\cdot\|$ denotes the $2$-norm of a vector and the induced $2$-norm of a matrix by default. 
A function $f: \bR^n \to \bR_+$ is positive definite if $f(x) \seq 0 \iff x \seq \bzero_n$. The comparison function classes $\cK$ and $\cKL$ are defined in the usual sense following~\cite{kellett2014compendium}. 
Given $x \in \bR^n$ and $Q \in \bR^{n\times n}$, $x^\top Qx$ is denoted by $\|x\|^2_Q$. The maximum (minimum) eigenvalue of a symmetric positive-definite matrix $M$ is denoted by $\eigmax{M}$ ($\eigmin{M}$).

%% file: 1.2-preliminaries/1.2.2-basic.tex
\vspace*{-0.7cm}
\subsection{System Description \& Optimal Control}
\label{subsec2.2-system_control}
Consider a discrete-time nonlinear system modeled by
\begin{equation}
    \label{eq:sec2-basic_model}
    x_{t+1} = \tsys{t}(x_t, u_t) = f(x_t, u_t) + \Df_t(x_t, u_t),
\end{equation}
where $x_t \in \cX = \bR^n$ and $u_t \in \cU \subseteq \bR^m$ are, respectively, the state and input at time step $t \in \bN$, $\tsys{t}: \cX\times\cU \to \cX$ ($t \in \bN$) is the \textit{time-varying} true model, $f: \cX\times\cU \to \cX$ is the nominal model, and $\Df_t:\cX\times\cU \to \cX$ ($t \in \bN$) describes the \textit{unknown} time-varying model mismatch, which is typical in robotics~\cite{wang2018safe, dacs2025robust} and automotive engineering~\cite{ames2016control, liu2025robust}. The input-constraint set $\cU$ is \textit{closed}.  
Consider the state $x_t$ at time step $t$ and an input sequence $\mathbf{v}_{0:n} = (v_0, \dots, v_{n}) \in \mathcal{U}^{n+1}$ of length $n+1 \in \mathbb{I}_{1:\infty}$. For $k \in \mathbb{I}_{0:n+1}$, let $\psi_{k|t}(\mathbf{v}_{0:n})$ and $\psi^\ast_{k|t}(\mathbf{v}_{0:n})$ denote the nominal and true states, respectively, that evolve according to $\psi_{k+1|t}(\mathbf{v}_{0:n}) = f(\psi_{k|t}(\mathbf{v}_{0:n}), v_k)$ and $\psi^\ast_{k+1|t}(\mathbf{v}_{0:n}) = f^\ast_{t+k}(\psi^\ast_{k|t}(\mathbf{v}_{0:n}), v_k)$, with the initial conditions $\psi_{0|t}(\mathbf{v}_{0:n}) = \psi^\ast_{0|t}(\mathbf{v}_{0:n}) = x_t$. In addition, the shorthand notations $\tsysn \coloneqq \{f^\ast_\tau\}^{t+N-1}_{\tau=t}$ and $\tsysc \coloneqq \{f^\ast_t\}^{\infty}_{t=0}$ are introduced for convenience. The following assumption summarizes the properties of the considered nonlinear system.
\begin{assumption}
    \label{ass:system_setting}
    The functions $f$ and $\{\Delta f_t\}^\infty_{t=0}$, along with the set $\mathcal{U}$, satisfy the following conditions:
    \begin{enumerate}[label=(\alph*), font=\upshape]
        \item $f(\bzero_n, \bzero_m) = \bzero_n$ and $\forall t \in \bN$, $\Df_t(\bzero_n,\bzero_m) = \bzero_n$;
        \item $u = \bzero_m$ lies in the interior of $\cU$;
        \item $f$ is Lipschitz continuous with respect to $x$ on $\cX$ uniformly over $\cU$, i.e., $\exists L_{f,x} \in (0, +\infty)$ such that $\forall u \in \cU$,
        \begin{equation}
            \label{eq:Lipschitz_nominal}
            \|f(x',u) - f(x'',u)\| \leq L_{f,x}\|x'-x''\|
        \end{equation}
        for any $x',x''\in\cX$;
        \item The functions $\{\Delta f_t\}_{t=0}^\infty$ are uniformly Lipschitz continuous on $\mathcal{X} \times \mathcal{U}$: for all $t \in \mathbb{N}$,
        \begin{multline}
            \label{eq:Lipschitz_uncertainty}
            \|\Delta f_t(x',u') - \Delta f_t(x'',u'')\| \\ \leq L_{\Delta f}(\mmf) \left( \|x' - x''\| + \|u' - u''\| \right)
        \end{multline}
        for any $x', x'' \in \mathcal{X}$ and $u', u'' \in \mathcal{U}$, where $\mmf \in [0, \bar{\varepsilon}_f]$, $L_{\Delta f}\hspace{-0.3em}:\hspace{-0.4em} \mathbb{R}_+ \sto \mathbb{R}_+$ is an increasing, continuous function satisfying $\lim_{\varepsilon \searrow 0} L_{\Delta f}(\varepsilon) = 0$, and $L_{\Delta f}(\bar{\varepsilon}_f) \eqqcolon \bar{L}_{\Delta f}$.
    \end{enumerate}
\end{assumption}
The condition (a), being widely satisfied in many applications (e.g.,~inverted pendulum~\cite{kohler2023stability, kuntz2024beyond} and robotics~\cite{wang2018safe, worthmann2015model}), ensures that $(\bzero, \bzero)$ is a shared equilibrium of both the nominal and true models, which then enables CE-MPC to stabilize the true system at the origin~\cite{liu2026certainty, schimperna2025data}. 
The condition (b) assumes that input constraints are inactive at $(\bzero, \bzero)$, which is standard in MPC formulations~\cite{rawlings2017model}. 
The condition (c) simplifies the analysis of state perturbations in nonlinear MPC~\cite{liu2026certainty, wabersich2022cautious}, and it accommodates many applications (e.g., robotics~\cite{worthmann2015model, wang2018safe, dacs2025robust} and power electronics~\cite{wang2020event}). Condition (d) requires the rate of change of $\Delta f_t$ to be bounded. Given that $\Delta f_t(\mathbf{0}_n, \mathbf{0}_m) = \mathbf{0}_n$, this requirement implies that $\|\Delta f_t(x, u)\| \leq L_{\Delta f}(\mmf)(\|x\| + \|u\|), \forall (x,u) \in \cX\times\cU$, yielding a bound on $\Delta f_t$ scaled by $\mmf$ via a Lipschitz-type condition, and $\mmf$ thus represents the mismatch level with a prior upper bound $\bar{\varepsilon}_f$. Such a formulation seamlessly accommodates both parametric modeling errors, where $\mmf$ is the norm of the difference between the nominal and true parameters~\cite{liu2026certainty} (e.g., systems modeled by nonlinear functions with known kernels and unknown coefficients), and operator-theoretic modeling frameworks~\cite{schimperna2025data}, where the nominal model is constructed using the Koopman operator, and $\mmf$ represents the corresponding residual approximation error.

\begin{remark}
Smoothness of the functions $f$ and $\Delta f_t$ is not required. This flexibility enables the framework to accommodate piecewise-affine (PWA) true systems under either static smooth or PWA nominal approximations, thereby naturally capturing a broad class of hybrid MPC problems~\cite{bemporad2002hybrid}.
\end{remark}

The control objective is to regulate the state to the origin asymptotically while minimizing the quadratic stage cost:
\begin{equation}
    \label{eq:quadratic_stage_cost}
    \ell(x, u) = \|x\|^2_Q + \|u\|^2_R,
\end{equation}
where $Q \succ 0$ and $R \succ 0$ are positive definite matrices\footnote{Although $R$ is conventionally allowed to be positive semi-definite, strict positive definiteness ($R \succ 0$) is commonly required in theoretical analyses to establish stability and performance guarantees~\cite{liu2024stability, li2025learning}.}. The optimized stage cost~\cite{grune2017nonlinear, kohler2023stability} follows as
\begin{equation}
    \label{eq:optimized_stage_cost}
    \els(x) := \min_{u \in \cU}\ell(x, u) = \ell(x,\bzero_m) = \|x\|^2_Q.
\end{equation}
Starting from an initial state $x_0 = x$, the oracle infinite-horizon optimal control solves the following optimization problem:
\begin{align*} 
	\probopt{x}{\tsysc}: & \min_{\{\nu_{k}, \xi_{k}\}^{\infty}_{k=0}} \sum^{\infty}_{k=0} \cost\left(\xi_{k}, \nu_{k}\right)
	\\ \text{s.t. } & \;\xi_{k+1} = \tsys{k}(\xi_{k}, \nu_{k}), \quad \forall k \in \bN;
	\\ & \; \nu_{k} \in \cU, \quad \forall k \in \bN;
	\\ & \; \xi_{0} = x,
\end{align*}
where $\xi_{k}$ and $\nu_{k}$ denote the predicted state and input at prediction step $k$, respectively, given the initial state $\xi_{0} = x$. In this work, $\probopt{x}{\tsysc}$ is assumed to be \textit{well-posed}, i.e., it admits a global minimizer $\bu^\star_{0:\infty}(x;\tsysc)$ and its associated optimal value $V_{\infty}(x;\tsysc) < +\infty$, implying that $\bu^\star_{0:\infty}(x;\tsysc)$ stabilizes the system at the origin asymptotically. Under input constraints, $\probopt{x}{\tsysc}$ may not be well-posed for an arbitrary $x$, and we therefore define the cost-induced region of attraction (ROA) of \eqref{eq:sec2-basic_model} as follows:
\begin{definition}[Cost-induced ROA~\cite{liu2026certainty}]
	\label{def:sec_pre:roa}
	The ROA of the system \eqref{eq:sec2-basic_model} for the stage cost $\ell$ in \eqref{eq:quadratic_stage_cost} is $\xroa(\ell) = \{x \in \cX \mid \exists \bu_{0:\infty} \in \cU^{\infty} \text{  s.t.}\;\; V_{\infty}(x;\tsysc) < +\infty\}$.
\end{definition}
It can be verified that $\xroa(\ell)$ is \textit{control invariant}~\cite{schimperna2025data, liu2026certainty}, and it is assumed that $x_0 \in \xroa(\ell)$; otherwise, the performance of IHOPC cannot be quantified. Characterizing $\xroa(\ell)$ without knowing $\tsysc$ is still an open challenge, and it is out of the scope of the current paper.

%% file: 1.2-preliminaries/1.2.3-st_mpc.tex
\vspace{-0.1cm}
\subsection{Asynchronous Model Predictive Control}
\label{subsec2.3-ampc}
Solving $\probopt{x}{\tsysc}$ is intractable for general nonlinear systems, and CE-MPC serves as a practical \textit{finite-horizon} surrogate using the nominal model. Specifically, at any time step $t$, the CE-MPC controller solves the following optimization problem:
\begin{align*} 
\probmpc{x_t}{f}: & \min_{\{\nu_{k|t}\}^{N-1}_{k=0}, \{\xi_{k|t}\}^{N}_{k=0}} \sum^{N-1}_{k=0} \cost(\xi_{k|t}, \nu_{k|t}) + \elf(\xi_{N|t}) 
\\ \text{s.t.  } & \;\xi_{k+1|t} = f(\xi_{k|t}, \nu_{k|t}), \quad \forall k \in \mathbb{I}_{0:N-1}; 
\\ & \; \nu_{k|t} \in \cU, \quad \forall k \in \mathbb{I}_{0:N-1}; 
\\ & \; \xi_{0|t} = x_t, 
\end{align*}
where $\xi_{k|t}$ and $\nu_{k|t}$ are, respectively, the predicted state and input at prediction step $k$ at time step $t$, $N \in \bN$ is the prediction horizon, and $\elf: \cX\to \bR_+$, defined by $\elf(x) = \|x\|^2_P$ with $P \succeq Q$, is the quadratic \textit{terminal} cost satisfying $\elf(x) \geq \els(x)$. Denote the optimal solution to $\probmpc{x_t}{f}$ by $\{\nu^\star_{k|t}(f)\}^{N-1}_{k=0}$ and $\{\xi^\star_{k|t}(f)\}^{N}_{k=0}$. In standard MPC with \textit{synchronous} feedback~\cite{rawlings2017model}, only the first optimal input $\nu^\star_{0|t}(f)$ is applied to the system, discarding the remaining ones $\{\nu^\star_{k|t}(f)\}^{N-1}_{k=1}$. Such \textit{time-triggered} control schemes require measuring the state $x_t$ and solving the problem $\probmpc{x_t}{f}$ at each time step $t$, which is not computationally efficient. The asynchronous CE-MPC scheme only solves the optimization problem $\probmpc{x_t}{f}$ at the triggering time steps $t_j \in \cT_M$, where the feedback time sequence $\cT_M$ is given by
\begin{equation}
    \label{eq:feedback_time_sequence}
    \cT_M = \{t_1, t_2, \dots, t_j, \dots \mid t_j \in \bN\}
\end{equation}
subject to $t_1 \seq 0$ and $t_{j}+1 \sleq t_{j+1} \sleq t_{j} + M$, with $M \leq N-1$ being the \textit{maximum} allowable inter-execution step. Consequently, the control loop is closed only at the time steps $t_j \in \cT_M$. The individual inter-execution steps are given as
\begin{equation*}
    m_j = t_{j+1} - t_j, \quad j \in \bN,
\end{equation*}
which satisfies $m_j \leq M$. Accordingly, the inputs applied at time steps $(t_j, t_j +1, \dots, t_{j+1}-1)$ are $\{\nu^\star_{k|t_j}(f)\}^{m_j-1}_{k=0}$. Note that the resulting MPC problem is the multi-step CE-MPC problem~\cite{grune2010analysis} when $\forall j \in \bN$, $m_j = m^\ast$, and it further degenerates to time-triggered CE-MPC~\cite{rawlings2017model, liu2026certainty} if $m^\ast=1$.

%% file: 1-major_content/3-problem_formulation.tex
\section{Problem Formulation}
\label{sec:3-problem_formulation}
Let $V_N(x;f)$ denote the optimal value of $\probmpc{x}{f}$, which implicitly defines the CE-MPC value function $V_N(\cdot;f): \cX \to \bR_+$ for a given horizon $N$. To analyze the resulting infinite-horizon performance, we introduce a \textit{cost controllability} condition on $V_N(\cdot;f)$ alongside a \textit{relaxed} control Lyapunov function (rCLF) property on the terminal cost $\elf$.
\begin{assumption}[Nominal Cost Controllability]
\label{ass:cost_controllability}
    There exists a uniform bound $\umpc$ such that for any prediction horizon $i \in \bI_{1:\infty}$, there exists a constant $\mpc_i \in (0, \umpc]$ such that $\forall x \in \xroa(\ell)$,
    \begin{equation}
        \label{eq:cost_controllability}
        V_i(x;f) \leq (1 + \mpc_i)\els(x),
    \end{equation}
    where $\els$ is the optimized stage cost defined in~\eqref{eq:optimized_stage_cost}.
\end{assumption}
\begin{assumption}[Terminal Cost]
    \label{ass:terminal_cost_relaxed_clf}
    The terminal cost function $\elf$ is a relaxed control Lyapunov function (rCLF), i.e., there exists a constant $\epf \in \bR_+$ such that $\forall x \in \xroa(\ell)$,
    \begin{equation}
        \label{eq:relaxed_clf}
        \min_{u\in\cU}\big\{\elf(f(x,u)) + \ell(x,u)\big\} \leq (1 + \epf)\els(x).
    \end{equation}
\end{assumption}
Assumptions~\ref{ass:cost_controllability} and~\ref{ass:terminal_cost_relaxed_clf} are standard, well-justified technical tools in MPC stability and performance analysis using relaxed dynamic programming~\cite{grune2008infinite, grune2010analysis, grune2017nonlinear, kohler2023stability, liu2026certainty}. Specifically, Assumption~\ref{ass:cost_controllability} holds locally within a compact neighborhood of the origin if the system possesses a controllable linearization at the origin and a strictly positive definite stage cost~\cite{grune2017nonlinear}. Inequality~\eqref{eq:cost_controllability} bounds the MPC value function relative to the stage cost, serving as an analytical bridge to characterize the infinite-horizon performance of MPC as a function of the prediction horizon~\cite{grune2008infinite, grune2017nonlinear}. This assumption is not restrictive; it represents a baseline well-posedness requirement, as any system violating it is inherently unstabilizable under a bounded cost. Because engineering systems are fundamentally designed to be locally controllable around their steady state, Assumption~\ref{ass:cost_controllability} has been used in mobile robots~\cite{worthmann2015model} and inverted pendulum~\cite{kohler2023stability}. Meanwhile, Assumption~\ref{ass:terminal_cost_relaxed_clf} serves to analyze the stability of the nominal system~\cite{kohler2023stability}, where the constant $\epf$ quantifies the degree of relaxation. Notably, \eqref{eq:relaxed_clf} recovers the standard CLF property~\cite{rawlings2017model} when $\epf=0$. Furthermore, Assumption~\ref{ass:terminal_cost_relaxed_clf} can be trivially satisfied by choosing a sufficiently large $\epf$~\cite{kohler2023stability}. In particular, $\frac{1+\gamma_1}{1+\epf} \in \left[\frac{\eigmin{P}}{\eigmax{Q}}, \frac{\eigmax{P}}{\eigmin{Q}}\right]$~\cite[Assumption 5]{kohler2023stability}, for which $\epf = \gamma_1$ offers a simple baseline. Crucially, because~\eqref{eq:cost_controllability} and~\eqref{eq:relaxed_clf} depend solely on the nominal dynamics $f$, they are completely decoupled from the unknown true plant $\tf_t$. Consequently, the parameters $\mpc_i$, $\epf$, and the uniform upper bound $\umpc$ can all be approximated numerically offline~\cite{kohler2023stability, grune2017nonlinear, schimperna2025data}.

To facilitate the subsequent analysis, we characterize the true state trajectory under the asynchronous CE-MPC framework at the \textit{open-loop} time steps $(t_j, t_j+1, \dots, t_{j+1})$, driven by the inputs $\{\nu^\star_{k|t_j}(f)\}^{m_j-1}_{k=0}$. For notational simplicity, for each $k \in \bI_{0:m_j}$, let $x_{j,k} \coloneqq x_{t_j+k}$ denote the true system state at time step $t_j+k$ within the $j$-th inter-execution range, and let $\nu^\star_{j,k}(f) \coloneqq \nu^\star_{k|t_j}(f)$ represent the applied input. The true state evolution within the $j$-th range is then governed by
\begin{equation}
    \label{eq:true_state_evolution}
    x_{j,k+1} = f^\ast_{t_j+k}(x_{j,k}, \nu^\star_{j,k}(f)),\;k \in \bI_{0:m_j-1},
\end{equation}
subject to the boundary condition $x_{j,m_j} = x_{j+1,0}$ matching consecutive ranges. 

Accordingly, given the feedback time sequence $\cT_M$ and the initial state $x_0 = x_{1,0} = x$, the infinite-horizon performance of the asynchronous CE-MPC controller is given by:
\begin{equation}
    \cJ^{[N,M]}_{\infty}(x;f) = \sum^\infty_{j=1}\sum^{m_j-1}_{k=0}\ell(x_{j,k+1}, \nu^\star_{j,k}(f)).
\end{equation}
The theoretical objectives of the current work are twofold: (i) to establish rigorous conditions under which the asynchronous CE-MPC scheme asymptotically stabilizes the true system at the true equilibrium $(\bzero_n, \bzero_m)$, ensuring that $\lim_{j\to\infty}x_{j,k} = \bzero_n$, and (ii) to derive an explicit upper bound on the asynchronous CE-MPC infinite-horizon performance $\cJ^{[N,M]}_{\infty}(x;f)$ relative to the ideal infinite-horizon performance $V_{\infty}(x;\tsysc)$, thus explicitly quantifying the suboptimality gap induced by the combined effects of a finite prediction horizon, plant-model mismatch, and the asynchronous feedback mechanism.

%% file: 1-major_content/4-theoretical_analysis.tex
\vspace*{-0.2cm}
\section{Analysis of Nominal Asynchronous Model Predictive Control}
\label{sec:4-theoretical_analysis}
In this section, we analyze the nominal stability and performance of asynchronous MPC under the assumption of perfect model alignment, i.e., $f^\ast_t = f$ ($\Df_t = \bzero_n$) for all $t \in \bN$. These results extend~\cite{grune2010analysis} by generalizing from synchronous multi-step feedback to arbitrary asynchronous feedback, while explicitly characterizing the stability conditions and performance bounds as functions of the maximum open-loop execution interval $M$. Furthermore, this analysis generalizes the frameworks in~\cite{liu2026certainty} and~\cite{kohler2023stability} by quantifying the effects of asynchronous feedback. While serving as an essential foundation for the subsequent robust analysis under plant-model mismatch in Section~\ref{sec:theoreticalAnalysis_model_mismatch}, these nominal results possess independent theoretical merit due to these distinct extensions.

\begin{theorem}[Nominal Stability and Performance]
    \label{thm:nominal_stability_performance}
    Let Assumptions~\ref{ass:system_setting}--\ref{ass:terminal_cost_relaxed_clf} hold. Define the stability index $\delta_{N,M}$ as
    \begin{equation}
        \label{eq:nominal_stability_margin}
        \delta_{N,M} := \frac{(1+\umpc)\mpc_N[(1 + \epf)^M-1]}{1+\umpc+(N-M)(1+\epf)^M},
    \end{equation}
    where $\umpc$ and $\mpc_N$ are given in Assumption~\ref{ass:cost_controllability}, and $\epf$ is given in Assumption~\ref{ass:terminal_cost_relaxed_clf}. If the prediction horizon $N$ and maximum execution interval $M$ satisfy $\delta_{N,M} < 1$, then:
    \begin{enumerate}
        \item[\normalfont(i)] The asynchronous closed-loop nominal system
        \begin{equation*}
            \label{eq:nominal_closed_loop}
            x_{j,k+1} = f(x_{j,k}, \nu^\star_{j,k}(f)),\;j \in \bI_{1:\infty},\;k \in \bI_{0:m_j-1},
        \end{equation*}
        is asymptotically stable at the origin under the MPC controller specified by $\probmpc{x_t}{f}$ and the feedback time sequence $\cT_M$ in~\eqref{eq:feedback_time_sequence}.
        \item[\normalfont(ii)] The infinite-horizon performance satisfies the following performance bound:
        \begin{equation}
            \label{eq:nominal_performance_bound}
            \cJ^{[N,M]}_{\infty}(x;f) \leq \underbrace{\frac{1 + \frac{\eigmax{P}}{\eigmin{Q}}\left(\frac{\umpc}{1+\umpc}\right)^N}{1 - \delta_{N,M}}}_{:= \Rcrnom} V_{\infty}(x;f),
        \end{equation}
        where $\Rcrnom$ is the nominal competitive-ratio bound.
    \end{enumerate}
\end{theorem}

The proof of Theorem~\ref{thm:nominal_stability_performance} is given in Appendix~\ref{appendix:A-proof-thm2}. For a fixed maximum step $M$, it holds that $\lim_{N \to \infty} \delta_{N,M} = 0$ and $\lim_{N \to \infty} \Rcrnom = 1$. This implies that a longer prediction horizon is beneficial to enhance system performance under perfect modeling; it further indicates that as the horizon approaches infinity, the suboptimality induced by the asynchronous feedback mechanism vanishes. Notably, by setting $M = 1$ (i.e., $m_j = 1$ for all $j \in \bI_{1:\infty}$), Theorem~\ref{thm:nominal_stability_performance} recovers the standard synchronous results in~\cite[Proposition 2]{liu2026certainty} and~\cite[Theorem 5]{kohler2023stability}. Conversely, for finite prediction horizon $N$, $\delta_{N,M}$ increases monotonically with $M$ (see Lemma~\ref{lm:increasing_function} in Appendix~\ref{appendix:A-lemmas}). This confirms that performance may degrade as the maximum inter-execution step expands, highlighting the necessity of frequent feedback to maintain performance when prediction capabilities are constrained~\cite[Chapter 10.4]{grune2017nonlinear}.

%% file: 1-major_content/4.5-main_results.tex
\section{Stability \& Performance Analysis Under Plant-Model Mismatch}
\label{sec:theoreticalAnalysis_model_mismatch}
This section extends the results in Section~\ref{sec:4-theoretical_analysis} to account for plant-model mismatch. First, Section~\ref{sec4.5-1:perturbation_analysis} analyzes cost and value function perturbations to quantify mismatch impacts. Leveraging this, Section~\ref{sec4.5-2:final_analysis} delivers our main contributions: an explicit closed-loop stability condition, a competitive-ratio performance bound, and an analysis pipeline. These results advance~\cite{liu2026certainty} by accommodating arbitrary asynchronous feedback, bypassing restrictive constraint qualifications, and permitting time-varying, non-smooth mismatch. All proofs are provided in Appendix~\ref{appendix:B}.

\input{1.4.5-main_results_analysis/1.4.5.1-perturbation}
\input{1.4.5-main_results_analysis/1.4.5.2-full_analysis}

%% file: 1.4.5-main_results_analysis/1.4.5.1-perturbation.tex
\subsection{Perturbation Analysis}
\label{sec4.5-1:perturbation_analysis}
To analyze the perturbation of the MPC value function under model mismatch, a bound on the \textit{open-loop} state perturbation is first established.

\begin{lemma}[State Perturbation] \label{lem:state_perturbation}
    Let Assumption~\ref{ass:system_setting} hold. For any time step $t$ and input sequence $\bv_{0:n} \in \cU^{n+1}$, the state perturbation satisfies
    \begin{multline} \label{eq:state_perturbation}
    \hspace{-1em}\|\psi^\ast_{k+1|t}(\bv_{0:n}) - \psi_{k+1|t}(\bv_{0:n})\| \\ 
    \hspace{-0.6em}\leq \sum^{k}_{i=0} \left[L^\ast_f(\mmf)\right]^{k-i} \|\Df_{t+i}(\psi_{i|t}(\bv_{0:n}), v_i)\|, \forall k \in \bI_{0:n}, \hspace{-0.4em}
    \end{multline}
    where $L^\ast_f(\mmf) \coloneqq L_f+ L_{\Df}(\mmf)$ with $L_f$ and $L_{\Df}(\mmf)$ given in Assumption~\ref{ass:system_setting}. In addition, the inequality~\eqref{eq:state_perturbation} also holds when $[L^\ast_f(\mmf)]^{k-i}\Df_{t+i}(\psi_{i|t}(\bv_{0:n}), v_i)$ is substituted with $(L_f)^{k-i}\Df_{t+i}(\psi^\ast_{i|t}(\bv_{0:n}), v_i)$, i.e., when the true state trajectory is employed in the upper bound.
\end{lemma}

To proceed with the perturbation analysis, the model mismatch is then bounded in terms of the stage cost via the \textit{consistent error-matching} property~\cite[Proposition 1]{liu2026certainty}. For quadratic costs, this property is formulated as follows:
\begin{lemma}[Consistent Error Matching]
    \label{lem:error_matching}
    Let Assumption~\ref{ass:system_setting} hold. Define $\mu_{\mathrm{ce}} := 2\left(\eigmin{Q}^{-1} + \eigmin{R}^{-1}\right)$. For all $t \in \bN$, the mismatch function $\Df_t$ satisfies
    \begin{equation}
        \label{eq:mismatch_bound}
        \hspace{-0.6em} \|\Df_t(x,u)\|^2 \leq \mu_{\mathrm{ce}}[L_{\Df}(\mmf)]^2 \ell(x,u), \forall (x,u) \in \cX\times\cU,
    \end{equation}
    where $L_{\Df}$ is given in Assumption~\ref{ass:system_setting}-(d).
\end{lemma}

Lemmas~\ref{lem:state_perturbation} and~\ref{lem:error_matching} establish the 
foundation for the perturbation analysis of the MPC value function. While a 
similar analysis for time-triggered feedback was systematically investigated in~\cite{liu2026certainty} via 
sensitivity under constraint qualifications, 
the current work develops an alternative approach tailored for asynchronous feedback with less restrictive assumptions.

Recall from Section~\ref{subsec2.3-ampc} that $\xi^\star_{k|t}(f) = \psi_{k|t}(\{\nu^\star_{\tau|t}(f)\}^{N-1}_{\tau=0})$ for $k \in \bI_{0:N}$, its associated true state, driven by $\{\nu^\star_{\tau|t}(f)\}^{N-1}_{\tau=0}$ under the true models $\tsysn$, is defined by $\eta^\ast_{k|t}(f) \coloneqq \psi^\ast_{k|t}(\{\nu^\star_{\tau|t}(f)\}^{N-1}_{\tau=0})$ for $k \in \bI_{0:N}$. For convenience, we compact the stage cost notation by defining $\ell_{k|t}(f) \coloneqq \ell(\xi^\star_{k|t}(f), \nu^\star_{k|t}(f))$ and $\ell^\ast_{k|t}(f) = \ell(\eta^\ast_{k|t}(f), \nu^\star_{k|t}(f))$. Analogous to the CE-MPC problem $\probmpc{x_t}{f}$, the true MPC problem $\probmpc{x_t}{\tsysn}$ is defined by governing state transitions with these true models. Following the definition of the CE-MPC value function $V_N(\cdot;f)$ in Section~\ref{sec:3-problem_formulation}, the optimal value of $\probmpc{x}{\tsysn}$, denoted by $V_N(x;\tsysn)$, defines the true MPC value function $V_N(\cdot;\tsysn): \cX \to \bR_+$.

The first perturbation result characterizes the state perturbation of the CE-MPC value function when evaluated at $\xi^\star_{k|t}(f)$ and $\eta_{k|t}(f^\ast)$, specifically quantifying the perturbation propagating directly from the underlying model mismatch.
\begin{proposition}[State-Driven Value-Function Perturbation]
    \label{prop:state_perturbation}
    Let Assumptions~\ref{ass:system_setting} and~\ref{ass:cost_controllability} hold. For each $k \in \bI_{1:N-1}$, there exists a function $\alpha_{N,k}:\bR_+ \to \bR_+$ such that for all $t \in \bN$,
    \begin{equation}
        \label{eq:state_perturbation_original}
        \hspace{-0.6em} V_N(\eta_{k|t}(f^\ast);f)\sminus V_N(\xi^\star_{k|t}(f);f) \leq \alpha_{N,k}(\mmf)\sum^{k-1}_{i=0}\ell_{i|t}(f),\hspace{-0.4em}
    \end{equation}
    where $\alpha_{N,k}(\cdot) = \alpha^{(1)}_{N,k}L_{\Df}(\cdot) + \alpha^{(2)}_{N,k}[L_{\Df}(\cdot)]^2$, with explicit expressions for the coefficients $\alpha^{(1)}_{N,k}$ and $\alpha^{(2)}_{N,k}$ provided in~\eqref{eq:appB.2-final_derviation_state_perturbation_bound} in Appendix~\ref{appendix:B-2-propositions}.
\end{proposition}
Subsequently, we characterize the functional perturbation induced by the model mismatch, defined as the discrepancy between the CE-MPC and true MPC value functions evaluated at an identical state.
\begin{proposition}[Model-Driven Value-Function Perturbation]
    \label{prop:model_mismatch}
    Let Assumptions~\ref{ass:system_setting} and~\ref{ass:cost_controllability} hold. For each $x \in \cX$, there exists a function $\beta_{N}:\bR_+ \to \bR_+$ such that for all $t \in \bN$,
    \begin{equation}
        \label{eq:model_mismatch_original}
        V_N(x;f)- V_N(x;\tsysn) \leq \beta_{N}(\mmf)V_N(x;\tsysn),
    \end{equation}
    where $\beta_{N}(\cdot) = \beta^{(1)}_{N}L_{\Df}(\cdot) + \beta^{(2)}_{N}[L_{\Df}(\cdot)]^2$, with explicit expressions for the coefficients $\beta^{(1)}_{N}$ and $\beta^{(2)}_{N}$ provided, respectively, in~\eqref{eq:appB.2-model_perturbation_expansion_derivation_final} and~\eqref{eq:appB.2-model_perturbation_expansion_derivation_quad} in Appendix~\ref{appendix:B-2-propositions}.
\end{proposition}

Proposition~\ref{prop:model_mismatch} bounds the difference between the two MPC value functions, accounting for both the stage and terminal costs along two distinct trajectories. In addition, the finite sum of the first $k$ stage costs for the state trajectories $\xi^\star_{k|t}(f)$ and $\eta_{k|t}(f^\ast)$ can also be isolated and bounded, excluding the terminal cost. This yields the following proposition, which is also instrumental in the subsequent analysis:
\begin{proposition}[Model-Driven $k$-Sum Perturbation]
    \label{prop:k_sum_mismatch}
    Let Assumptions~\ref{ass:system_setting} and~\ref{ass:cost_controllability} hold. For all $k \in \bI_{1:N-1}$, there exists a function $\theta_{k}:\bR_+ \to \bR_+$ such that for all $t \in \bN$,
    \begin{equation}
        \label{eq:model_mismatch_k_sum}
        \left|\sum^{k-1}_{i=0}\ell^\ast_{i|t}(f) - \sum^{k-1}_{i=0}\ell_{i|t}(f)\right|
        \leq \theta_k(\mmf)\sum^{k-1}_{i=0}\ell_{i|t}(f)
    \end{equation}
    where $\theta_{k}(\cdot) = \theta^{(1)}_{k}L_{\Df}(\cdot) + \theta^{(2)}_{k}[L_{\Df}(\cdot)]^2$. The explicit expressions for the coefficients $\theta^{(1)}_{k}$ and $\theta^{(2)}_{k}$ are provided in~\eqref{eq:appB.2-k_sum_expansion_derivation_final} in Appendix~\ref{appendix:B-2-propositions}.
\end{proposition}

In summary, these results establish the fundamental bounds required to quantify the effects of model mismatch on the MPC value functions. Proposition~\ref{prop:state_perturbation} bounds the perturbation of the CE-MPC value function induced by state trajectory deviations, while Proposition~\ref{prop:model_mismatch} quantifies the point-wise functional discrepancy between the CE-MPC and true MPC value functions. Complementing these, Proposition~\ref{prop:k_sum_mismatch} leverages the proof technique of Proposition~\ref{prop:model_mismatch} to bound the accumulated model-mismatch error on the stage costs over a finite horizon, excluding the terminal cost. This suite of results establishes a rigorous foundation for the subsequent closed-loop stability and suboptimality analysis.

%% file: 1.4.5-main_results_analysis/1.4.5.2-full_analysis.tex
\subsection{Competitive-Ratio Performance Bound}
\label{sec4.5-2:final_analysis}
In essence, the stability of the nominal system under nominal asynchronous MPC (cf.~Theorem~\ref{thm:nominal_stability_performance}) is established by employing $V_N(\cdot;f)$ as a Lyapunov function candidate satisfying a standard dissipation inequality, consistent with conventional MPC stability proofs~\cite{grune2017nonlinear, rawlings2017model}. However, $V_N(\cdot;f)$ guarantees stability only under nominal state trajectories. To extend this to the closed-loop system~\eqref{eq:true_state_evolution} under asynchronous CE-MPC, nominal stability must be coupled with the perturbation bounds in Propositions~\ref{prop:state_perturbation} and~\ref{prop:model_mismatch}. The stability and the associated competitive-ratio performance bound are given in the following theorem:
\begin{theorem}[Stability and Performance of Asynchronous CE-MPC]
    \label{thm:stability_performance_acempc}
    Let Assumptions~\ref{ass:system_setting}--\ref{ass:terminal_cost_relaxed_clf} hold. Given the nominal stability index $\delta_{N,M}$ in~\eqref{eq:nominal_stability_margin}, define the mismatch-driven stability index of asynchronous CE-MPC by 
    \begin{equation}
        \label{eq:stabiltity_index_acempc}
        \rho_{N,M}(\mmf) \coloneqq \delta_{N,M} + \alpha_{N,M}(\mmf),
    \end{equation}
    where $\alpha_{N}$ is given in Proposition~\ref{prop:state_perturbation}. If the prediction horizon $N$, the maximum inter-execution step $M$, and the mismatch level $\mmf$ jointly satisfy $\rho_{N,M}(\mmf) < 1$, then:
    \begin{enumerate}
        \item[\normalfont(i)] The asynchronous closed-loop system~\eqref{eq:true_state_evolution} is asymptotically stable at the origin under the asynchronous CE-MPC controller specified by $\probmpc{x_t}{f}$ and the feedback time sequence $\cT_M$ in~\eqref{eq:feedback_time_sequence}.
        \item[\normalfont(ii)] The infinite-horizon performance satisfies the following performance bound:
        \begin{equation}
            \label{eq:performance_bound_acempc}
            \cJ^{[N,M]}_{\infty}(x;f) \leq \Rcr V_{\infty}(x;\tsysc),
        \end{equation}
        where the competitive-ratio bound $\Rcr$ is given by
        \begin{multline}
            \label{eq:ratio_acempc_detail}
            \hspace*{-1em}\Rcr = \big[1 - \rho_{N,M}(\mmf)\big]^{-1}\big[1+\min\{\umpc,\beta_{N}(\mmf)\}\big]\\\big[1+\theta_M(\mmf)\big]\left[1 + \frac{\eigmax{P}}{\eigmin{Q}}\left(\frac{\bar{\gamma}^\ast(\mmf)}{1+\bar{\gamma}^\ast(\mmf)}\right)^N\right],
        \end{multline}
        with $\bar{\gamma}^\ast(\mmf) \coloneqq \umpc+\beta_{\infty}(\mmf)+\umpc\beta_{\infty}(\mmf)$, and $\beta_{N}$($\beta_{\infty}$) and $\theta_M$ defined as in Propositions~\ref{prop:model_mismatch} and~\ref{prop:k_sum_mismatch}, respectively.
    \end{enumerate}
\end{theorem}
The presence of $\alpha_{N,M}(\mmf)$ implies that having a larger model mismatch reduces the \textit{stability margin} $1 - \rho_{N,M}(\mmf)$ and increases $\Rcr$ (i.e., the worst-case performance bound). The bound~\eqref{eq:performance_bound_acempc} quantifies this degradation and remains consistent, recovering the ideal case $\Rcr = 1$ when $\mmf=0$ and $N=\infty$. Furthermore, because both $\alpha_{N,M}(\mmf)$ and $\theta_M(\mmf)$ increase with $M$ (see Appendix~\ref{appendix:B}), model mismatch degrades the stability margin and inflates the worst-case performance bound, especially for larger $M$. Thus, maintaining a target performance level under mismatch necessitates more frequent feedback than the nominal case.

On the other hand, because $\delta_{N,M}$ decreases as the prediction horizon $N$ increases, neither $\rho_{N,M}(\mmf)$ nor $\Rcr$ is monotonic with respect to $N$. Under model mismatch, a true optimal horizon may exist for certain systems, and the horizon derived from minimizing the competitive ratio provides a coarse warm start for tuning~\cite{liu2024stability, liu2026certainty}. Yet, in practice, system performance remains largely insensitive to horizon variations, particularly under infinitesimal mismatch~\cite{liu2026certainty}. Nonetheless, the stability condition $\rho_{N,M}(\mmf) < 1$ remains a valuable tool to derive a permissible range of $N$ given $M$ and $\mmf$, or to bound either $M$ or $\mmf$ given the other two parameters. Typically, the bound on $\mmf$ has a closed-form expression:
\begin{equation*}
    \label{eq:bound_model_mismatch}
    \mmf < L_{\Df}^{-1}\Bigg(\frac{-\alpha^{(1)}_{N,M} \splus \left[\big(\alpha^{(2)}_{N,M}\big)^2 \splus 4\alpha^{(2)}_{N,M}(1 \sminus \delta_{N,M})\right]^{\frac{1}{2}}}{2\alpha^{(2)}_{N,M}}\Bigg),
\end{equation*}
where $L_{\Df}^{-1}$ is the inverse function of $L_{\Df}$ given in Assumption~\ref{ass:system_setting}-(d), $\delta_{N,M}$ is the nominal stability index given in~\eqref{eq:nominal_stability_margin}, and $\alpha^{(1)}_{N,M}$ and $\alpha^{(2)}_{N,M}$ are given in Proposition~\ref{prop:state_perturbation}.

\subsection{Analysis Pipeline}
\label{subsec4.5-3:pipeline}
\input{figures/pipeline}
Inspired by the analysis pipeline in~\cite{liu2026certainty} for time-triggered CE-MPC without a terminal cost, we propose an extended analysis pipeline tailored for asynchronous CE-MPC with general terminal costs characterized by rCLFs, as illustrated in Fig.~1. The pipeline begins with the \textbf{Basic Lemmas}, where Lemma~\ref{lem:state_perturbation} (State Perturbation) and Lemma~\ref{lem:error_matching} (Consistent Error Matching) provide the foundational results for the subsequent \textbf{Perturbation Analysis}. Specifically, these lemmas jointly drive the derivation of Proposition~\ref{prop:state_perturbation} (State-Driven Value-Function Perturbation) and Proposition~\ref{prop:model_mismatch} (Model-Driven Value-Function Perturbation), and the proof technique of Proposition~\ref{prop:model_mismatch} can be used directly to establish the model-driven $k$-sum perturbation bound in Proposition~\ref{prop:k_sum_mismatch}. In parallel, the nominal properties of \textbf{Asynchronous MPC} are established via Theorem~\ref{thm:nominal_stability_performance}, progressing from nominal stability to nominal performance. Finally, the perturbation bounds from Propositions~\ref{prop:state_perturbation}--\ref{prop:k_sum_mismatch} and the nominal guarantees from Theorem~1 are synthesized to establish the robust stability and performance of \textbf{Asynchronous CE-MPC} in Theorem~\ref{thm:stability_performance_acempc}.

Note that the stability condition $\rho_{N,M}(\delta_f) < 1$ is only sufficient, and the performance bound \eqref{eq:performance_bound_acempc} is inherently conservative. However, such conservatism is a recognized characteristic of theoretical performance analysis in the literature~\cite{grune2017nonlinear, liu2026certainty}. To minimize conservatism under stronger system assumptions and/or more structured model mismatch, the proposed analysis pipeline can be refined by deriving a tighter stability index $\delta_{N,M}$ and sharpening the bounding functions (i.e., $\alpha_{N,M}$, $\beta_{N}$ and $\theta_k$) in the results of value-function perturbation~\cite{liu2026certainty}. These refined analyses propagate directly through our pipeline to yield less conservative performance bounds.
\begin{remark}
    The competitive ratio can be computed in practice as it depends only on the mismatch level $\mmf$ and the nominal model $f$, independent of the inaccessible true models $\tsysc$. Notably, the pipeline and the results remain valid if the time-varying nature is inverted (i.e., a time-varying nominal model with a static true model), capturing scenarios in learning-based MPC with online model learning where the estimated model fluctuates near the true one.
\end{remark}

%% file: figures/pipeline.tex
\begin{figure}[h]
\label{fig:pipeline}
\centering
\resizebox{\linewidth}{!}{

\tikzset{every picture/.style={line width=0.75pt}} 

\begin{tikzpicture}[x=0.75pt,y=0.75pt,yscale=-1,xscale=1]

\draw  [color={rgb, 255:red, 144; green, 19; blue, 254 }  ,draw opacity=0.6 ][fill={rgb, 255:red, 144; green, 19; blue, 254 }  ,fill opacity=0.1 ][dash pattern={on 4.5pt off 4.5pt}] (83,22.57) .. controls (83,16.18) and (88.18,11) .. (94.57,11) -- (232.63,11) .. controls (239.02,11) and (244.2,16.18) .. (244.2,22.57) -- (244.2,107.23) .. controls (244.2,113.62) and (239.02,118.8) .. (232.63,118.8) -- (94.57,118.8) .. controls (88.18,118.8) and (83,113.62) .. (83,107.23) -- cycle ;
\draw   (87.96,40.2) .. controls (87.96,36.11) and (91.27,32.8) .. (95.36,32.8) -- (229.8,32.8) .. controls (233.89,32.8) and (237.2,36.11) .. (237.2,40.2) -- (237.2,62.4) .. controls (237.2,66.49) and (233.89,69.8) .. (229.8,69.8) -- (95.36,69.8) .. controls (91.27,69.8) and (87.96,66.49) .. (87.96,62.4) -- cycle ;
\draw   (88.2,83.31) .. controls (88.2,79.37) and (91.39,76.19) .. (95.32,76.19) -- (230.08,76.19) .. controls (234.01,76.19) and (237.2,79.37) .. (237.2,83.31) -- (237.2,104.68) .. controls (237.2,108.61) and (234.01,111.8) .. (230.08,111.8) -- (95.32,111.8) .. controls (91.39,111.8) and (88.2,108.61) .. (88.2,104.68) -- cycle ;
\draw  [color={rgb, 255:red, 245; green, 166; blue, 35 }  ,draw opacity=0.6 ][fill={rgb, 255:red, 245; green, 166; blue, 35 }  ,fill opacity=0.1 ][dash pattern={on 4.5pt off 4.5pt}] (274.2,27.26) .. controls (274.2,17.73) and (281.93,10) .. (291.46,10) -- (499.94,10) .. controls (509.47,10) and (517.2,17.73) .. (517.2,27.26) -- (517.2,153.54) .. controls (517.2,163.07) and (509.47,170.8) .. (499.94,170.8) -- (291.46,170.8) .. controls (281.93,170.8) and (274.2,163.07) .. (274.2,153.54) -- cycle ;
\draw   (281.2,40) .. controls (281.2,36.02) and (284.42,32.8) .. (288.4,32.8) -- (503,32.8) .. controls (506.98,32.8) and (510.2,36.02) .. (510.2,40) -- (510.2,61.6) .. controls (510.2,65.58) and (506.98,68.8) .. (503,68.8) -- (288.4,68.8) .. controls (284.42,68.8) and (281.2,65.58) .. (281.2,61.6) -- cycle ;
\draw   (281.2,82) .. controls (281.2,78.02) and (284.42,74.8) .. (288.4,74.8) -- (503,74.8) .. controls (506.98,74.8) and (510.2,78.02) .. (510.2,82) -- (510.2,103.6) .. controls (510.2,107.58) and (506.98,110.8) .. (503,110.8) -- (288.4,110.8) .. controls (284.42,110.8) and (281.2,107.58) .. (281.2,103.6) -- cycle ;
\draw   (281.2,135) .. controls (281.2,131.02) and (284.42,127.8) .. (288.4,127.8) -- (503,127.8) .. controls (506.98,127.8) and (510.2,131.02) .. (510.2,135) -- (510.2,156.6) .. controls (510.2,160.58) and (506.98,163.8) .. (503,163.8) -- (288.4,163.8) .. controls (284.42,163.8) and (281.2,160.58) .. (281.2,156.6) -- cycle ;
\draw  [color={rgb, 255:red, 139; green, 87; blue, 42 }  ,draw opacity=0.6 ][fill={rgb, 255:red, 139; green, 87; blue, 42 }  ,fill opacity=0.1 ][dash pattern={on 4.5pt off 4.5pt}] (84,140.18) .. controls (84,132.9) and (89.9,127) .. (97.18,127) -- (232.02,127) .. controls (239.3,127) and (245.2,132.9) .. (245.2,140.18) -- (245.2,236.62) .. controls (245.2,243.9) and (239.3,249.8) .. (232.02,249.8) -- (97.18,249.8) .. controls (89.9,249.8) and (84,243.9) .. (84,236.62) -- cycle ;
\draw   (88.96,155.2) .. controls (88.96,151.11) and (92.27,147.8) .. (96.36,147.8) -- (230.8,147.8) .. controls (234.89,147.8) and (238.2,151.11) .. (238.2,155.2) -- (238.2,177.4) .. controls (238.2,181.49) and (234.89,184.8) .. (230.8,184.8) -- (96.36,184.8) .. controls (92.27,184.8) and (88.96,181.49) .. (88.96,177.4) -- cycle ;
\draw   (89.96,213.2) .. controls (89.96,209.11) and (93.27,205.8) .. (97.36,205.8) -- (231.8,205.8) .. controls (235.89,205.8) and (239.2,209.11) .. (239.2,213.2) -- (239.2,235.4) .. controls (239.2,239.49) and (235.89,242.8) .. (231.8,242.8) -- (97.36,242.8) .. controls (93.27,242.8) and (89.96,239.49) .. (89.96,235.4) -- cycle ;
\draw  [color={rgb, 255:red, 65; green, 117; blue, 5 }  ,draw opacity=0.6 ][fill={rgb, 255:red, 65; green, 117; blue, 5 }  ,fill opacity=0.1 ][dash pattern={on 4.5pt off 4.5pt}] (274.2,187.49) .. controls (274.2,183.36) and (277.56,180) .. (281.69,180) -- (509.71,180) .. controls (513.84,180) and (517.2,183.36) .. (517.2,187.49) -- (517.2,242.31) .. controls (517.2,246.44) and (513.84,249.8) .. (509.71,249.8) -- (281.69,249.8) .. controls (277.56,249.8) and (274.2,246.44) .. (274.2,242.31) -- cycle ;
\draw   (285.96,213.2) .. controls (285.96,209.11) and (289.27,205.8) .. (293.36,205.8) -- (368.8,205.8) .. controls (372.89,205.8) and (376.2,209.11) .. (376.2,213.2) -- (376.2,235.4) .. controls (376.2,239.49) and (372.89,242.8) .. (368.8,242.8) -- (293.36,242.8) .. controls (289.27,242.8) and (285.96,239.49) .. (285.96,235.4) -- cycle ;
\draw   (414.96,212.4) .. controls (414.96,208.31) and (418.27,205) .. (422.36,205) -- (497.8,205) .. controls (501.89,205) and (505.2,208.31) .. (505.2,212.4) -- (505.2,234.6) .. controls (505.2,238.69) and (501.89,242) .. (497.8,242) -- (422.36,242) .. controls (418.27,242) and (414.96,238.69) .. (414.96,234.6) -- cycle ;
\draw [color={rgb, 255:red, 144; green, 19; blue, 254 }  ,draw opacity=1 ]   (237,45) -- (279.2,45) ;
\draw [shift={(281.2,45)}, rotate = 180] [fill={rgb, 255:red, 144; green, 19; blue, 254 }  ,fill opacity=1 ][line width=0.08]  [draw opacity=0] (9.6,-2.4) -- (0,0) -- (9.6,2.4) -- cycle    ;
\draw [color={rgb, 255:red, 144; green, 19; blue, 254 }  ,draw opacity=1 ]   (237,100.6) -- (279.2,100.6) ;
\draw [shift={(281.2,100.6)}, rotate = 180] [fill={rgb, 255:red, 144; green, 19; blue, 254 }  ,fill opacity=1 ][line width=0.08]  [draw opacity=0] (9.6,-2.4) -- (0,0) -- (9.6,2.4) -- cycle    ;
\draw [color={rgb, 255:red, 144; green, 19; blue, 254 }  ,draw opacity=1 ]   (237.2,91) -- (278.65,57.26) ;
\draw [shift={(280.2,56)}, rotate = 140.86] [fill={rgb, 255:red, 144; green, 19; blue, 254 }  ,fill opacity=1 ][line width=0.08]  [draw opacity=0] (9.6,-2.4) -- (0,0) -- (9.6,2.4) -- cycle    ;
\draw [color={rgb, 255:red, 144; green, 19; blue, 254 }  ,draw opacity=1 ]   (237.2,54) -- (278.65,87.74) ;
\draw [shift={(280.2,89)}, rotate = 219.14] [fill={rgb, 255:red, 144; green, 19; blue, 254 }  ,fill opacity=1 ][line width=0.08]  [draw opacity=0] (9.6,-2.4) -- (0,0) -- (9.6,2.4) -- cycle    ;
\draw [color={rgb, 255:red, 245; green, 166; blue, 35 }  ,draw opacity=1 ]   (390.2,111) -- (390.2,126) ;
\draw [shift={(390.2,128)}, rotate = 270] [fill={rgb, 255:red, 245; green, 166; blue, 35 }  ,fill opacity=1 ][line width=0.08]  [draw opacity=0] (9.6,-2.4) -- (0,0) -- (9.6,2.4) -- cycle    ;
\draw [color={rgb, 255:red, 139; green, 87; blue, 42 }  ,draw opacity=1 ]   (163.2,185) -- (163.2,204) ;
\draw [shift={(163.2,206)}, rotate = 270] [fill={rgb, 255:red, 139; green, 87; blue, 42 }  ,fill opacity=1 ][line width=0.08]  [draw opacity=0] (9.6,-2.4) -- (0,0) -- (9.6,2.4) -- cycle    ;
\draw [color={rgb, 255:red, 65; green, 117; blue, 5 }  ,draw opacity=1 ]   (376,224) -- (413.2,224) ;
\draw [shift={(415.2,224)}, rotate = 180] [fill={rgb, 255:red, 65; green, 117; blue, 5 }  ,fill opacity=1 ][line width=0.08]  [draw opacity=0] (9.6,-2.4) -- (0,0) -- (9.6,2.4) -- cycle    ;
\draw [color={rgb, 255:red, 245; green, 166; blue, 35 }  ,draw opacity=1 ]   (510,91) -- (529.2,91) ;
\draw [color={rgb, 255:red, 245; green, 166; blue, 35 }  ,draw opacity=1 ]   (507,224) -- (529.2,224) ;
\draw [shift={(505,224)}, rotate = 0] [fill={rgb, 255:red, 245; green, 166; blue, 35 }  ,fill opacity=1 ][line width=0.08]  [draw opacity=0] (9.6,-2.4) -- (0,0) -- (9.6,2.4) -- cycle    ;
\draw [color={rgb, 255:red, 245; green, 166; blue, 35 }  ,draw opacity=1 ]   (529.2,91) -- (529.2,224) ;
\draw [color={rgb, 255:red, 139; green, 87; blue, 42 }  ,draw opacity=1 ]   (238,166.6) -- (248.2,166.6) ;
\draw [color={rgb, 255:red, 139; green, 87; blue, 42 }  ,draw opacity=1 ]   (247.76,235.4) -- (283.96,235.4) ;
\draw [shift={(285.96,235.4)}, rotate = 180] [fill={rgb, 255:red, 139; green, 87; blue, 42 }  ,fill opacity=1 ][line width=0.08]  [draw opacity=0] (9.6,-2.4) -- (0,0) -- (9.6,2.4) -- cycle    ;
\draw [color={rgb, 255:red, 139; green, 87; blue, 42 }  ,draw opacity=1 ]   (248.2,166) -- (248.2,235) ;
\draw [color={rgb, 255:red, 245; green, 166; blue, 35 }  ,draw opacity=1 ]   (254,52.6) -- (280.2,52.6) ;
\draw [color={rgb, 255:red, 245; green, 166; blue, 35 }  ,draw opacity=1 ]   (254.2,52) -- (254.2,225) ;
\draw [color={rgb, 255:red, 245; green, 166; blue, 35 }  ,draw opacity=1 ]   (254,224.6) -- (284.2,224.6) ;
\draw [shift={(286.2,224.6)}, rotate = 180] [fill={rgb, 255:red, 245; green, 166; blue, 35 }  ,fill opacity=1 ][line width=0.08]  [draw opacity=0] (9.6,-2.4) -- (0,0) -- (9.6,2.4) -- cycle    ;
\draw [color={rgb, 255:red, 245; green, 166; blue, 35 }  ,draw opacity=1 ]   (265,214.6) -- (284.2,214.6) ;
\draw [shift={(286.2,214.6)}, rotate = 180] [fill={rgb, 255:red, 245; green, 166; blue, 35 }  ,fill opacity=1 ][line width=0.08]  [draw opacity=0] (9.6,-2.4) -- (0,0) -- (9.6,2.4) -- cycle    ;
\draw [color={rgb, 255:red, 245; green, 166; blue, 35 }  ,draw opacity=1 ]   (265,145.6) -- (281.2,145.6) ;
\draw [color={rgb, 255:red, 245; green, 166; blue, 35 }  ,draw opacity=1 ]   (265,145.6) -- (265,214.6) ;

\draw (128,15) node [anchor=north west][inner sep=0.75pt]  [font=\footnotesize,color={rgb, 255:red, 144; green, 19; blue, 254 }  ,opacity=1 ] [align=left] {\textbf{Basic Lemmas}};
\draw (90,37) node [anchor=north west][inner sep=0.75pt]  [font=\small] [align=left] {\begin{minipage}[lt]{108.84pt}\setlength\topsep{0pt}
\begin{center}
\textcolor[rgb]{0.82,0.01,0.11}{State} Perturbation\\(Lemma~\ref{lem:state_perturbation})
\end{center}

\end{minipage}};
\draw (90.17,81.33) node [anchor=north west][inner sep=0.75pt]  [font=\small] [align=left] {\begin{minipage}[lt]{108.84pt}\setlength\topsep{0pt}
\begin{center}
Consistent Error Matching\\
(Lemma~\ref{lem:error_matching})
\end{center}

\end{minipage}};
\draw (340,14) node [anchor=north west][inner sep=0.75pt]  [font=\footnotesize,color={rgb, 255:red, 245; green, 166; blue, 35 }  ,opacity=1 ] [align=left] {\textbf{Perturbation Analysis}};
\draw (283.4,36.8) node [anchor=north west][inner sep=0.75pt]  [font=\small] [align=left] {\begin{minipage}[lt]{169.92pt}\setlength\topsep{0pt}
\begin{center}
\textcolor[rgb]{0.82,0.01,0.11}{State-Driven} Value-Function Perturbation\\
(Proposition~\ref{prop:state_perturbation})
\end{center}

\end{minipage}};
\draw (282,79) node [anchor=north west][inner sep=0.75pt]  [font=\small] [align=left] {\begin{minipage}[lt]{173.49pt}\setlength\topsep{0pt}
\begin{center}
\textcolor[rgb]{0.82,0.01,0.11}{Model-Driven} Value-Function Perturbation\\
(Proposition~\ref{prop:model_mismatch})
\end{center}

\end{minipage}};
\draw (303.33,132.67) node [anchor=north west][inner sep=0.75pt]  [font=\small] [align=left] {\begin{minipage}[lt]{138.94pt}\setlength\topsep{0pt}
\begin{center}
\textcolor[rgb]{0.82,0.01,0.11}{Model-Driven} k-Sum Perturbation\\
(Proposition~\ref{prop:k_sum_mismatch})
\end{center}

\end{minipage}};
\draw (118.67,153.5) node [anchor=north west][inner sep=0.75pt]  [font=\small] [align=left] {\begin{minipage}[lt]{71.09pt}\setlength\topsep{0pt}
\textcolor[rgb]{0.82,0.01,0.11}{Nominal} Stability
\begin{center}
(Theorem~\ref{thm:nominal_stability_performance})
\end{center}

\end{minipage}};
\draw (116,131) node [anchor=north west][inner sep=0.75pt]  [font=\footnotesize,color={rgb, 255:red, 139; green, 87; blue, 42 }  ,opacity=1 ] [align=left] {\textbf{Asynchronous MPC}};
\draw (102.36,210.8) node [anchor=north west][inner sep=0.75pt]  [font=\small] [align=left] {\begin{minipage}[lt]{91.5pt}\setlength\topsep{0pt}
\begin{center}
\textcolor[rgb]{0.82,0.01,0.11}{Nominal} Performance\\(Theorem~\ref{thm:nominal_stability_performance})
\end{center}

\end{minipage}};
\draw (333,186) node [anchor=north west][inner sep=0.75pt]  [font=\footnotesize,color={rgb, 255:red, 65; green, 117; blue, 5 }  ,opacity=1 ] [align=left] {\textbf{Asynchronous CE-MPC}};
\draw (295.67,211) node [anchor=north west][inner sep=0.75pt]  [font=\small] [align=left] {\begin{minipage}[lt]{53.23pt}\setlength\topsep{0pt}
\begin{center}
\textcolor[rgb]{0.82,0.01,0.11}{CE} Stability\\
(Theorem~\ref{thm:stability_performance_acempc})
\end{center}

\end{minipage}};
\draw (408.36,210) node [anchor=north west][inner sep=0.75pt]  [font=\small] [align=left] {\begin{minipage}[lt]{75.28pt}\setlength\topsep{0pt}
\begin{center}
\textcolor[rgb]{0.82,0.01,0.11}{CE} Performance\\
(Theorem~\ref{thm:stability_performance_acempc})
\end{center}

\end{minipage}};
\end{tikzpicture}
}
\caption{Pipeline for performance analysis of asynchronous certainty-equivalence model predictive control (CE-MPC).}
\end{figure}
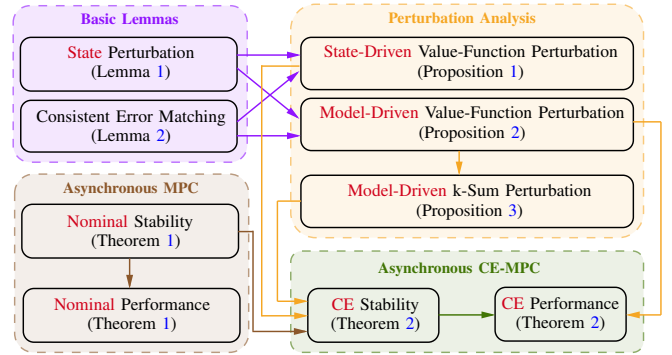

%% file: 1-major_content/5-case_study.tex
\vspace{-0.3cm}
\section{Numerical Example}
\label{sec:5-examples}
In this section, a numerical example is presented to verify the validity of the derived performance bound. Consider the following second-order discrete-time nonlinear system:
\begin{equation}
    \label{eq:numerical_example_model}
    \hspace{-0.5em}
    \begin{bmatrix}
        x^{(1)}_{t+1} \\
        x^{(2)}_{t+1}
    \end{bmatrix}
    \seq
    \underbrace{\begin{bmatrix}
        -0.6\sin(x^{(2)}_{t}) \\
        0.67(x^{(1)}_{t} \splus x^{(2)}_{t}) \splus u_t
    \end{bmatrix}}_{\text{nominal }f(x_t, u_t)}
    +
    \underbrace{\sin(\omega_t)
    \begin{bmatrix}
        |x^{(1)}_{t}| \\
        \tanh(u_t)
    \end{bmatrix}}_{\text{mismatch }\Delta f_t(x_t,u_t)},
\end{equation}
where $\omega_t = \mmf[1 + \cos(t)]$ reflects the time-varying nature of the model mismatch, with $\mmf \in [0, 0.01]$ denoting the mismatch level. System~\eqref{eq:numerical_example_model} readily satisfies Assumption~\ref{ass:system_setting} with $L_f = 0.95$, $L_{\Delta}(\mmf) = 2\mmf$, and $\bar{L}_{\Delta f} = 0.03$. The stage-cost matrices are set to $Q = \mathbf{I}_2$ and $R = 1$, the terminal-cost matrix is $P = 5\mathbf{I}_2$, and the input-constraint set is given by $\cU = \{u \in \bR \mid -0.1 \leq u \leq 0.1\}$. Due to the non-smooth mismatch term, the IHOPC problem with the true model is formulated as a mixed-integer program over a sufficiently long horizon and solved via Gurobi~\cite{gurobi}. Conversely, for the CE-MPC problem with the nominal model, the resulting nonlinear program is solved using IPOPT via CasADi~\cite{andersson2019}. All simulations are carried out in Python 3.13.11. The true performance ratio is defined as
\begin{equation}
    \label{eq:true_performance_ratio}
    R^{\diamond}_{N,M}(\mmf) = \frac{\cJ^{[N,M]}_{\infty}(x;f)}{V_{\infty}(x;\tsysc)},
\end{equation}
which is evaluated against the theoretical competitive-ratio bound $\Rcr$ in~\eqref{eq:ratio_acempc_detail}.

For the numerical setup, the prediction horizon is selected as $N = 25$. To guarantee that the stability condition $\rho_{N,M}(\mmf) < 1$ holds for all $\mmf \leq 0.01$, the considered maximum inter-execution step is restricted to $M \leq 8$, and the initial state is chosen from $x_0 \in \{[3,3]^\top, [-3,3]^\top, [3,-3]^\top, [-3,-3]^\top\}$. To encompass both general asynchronous feedback and multi-step MPC, two execution scenarios are evaluated for any given $M$: (i) \textbf{random inter-execution intervals}, where $m_j$ is drawn uniformly at random from $\bI_{1:M}$ over $100$ independent Monte Carlo runs, and (ii) \textbf{a fixed inter-execution interval}, where $m_j = M$. 

System performance is analyzed under two distinct experimental settings:
\begin{enumerate}
    \item \textbf{Varying the Maximum Inter-Execution Step $M$:} The mismatch level is fixed to $\mmf = 0.01$ while $M$ is varied over $M \in \bI_{2:8}$ under both scenarios. The simulation results are shown in Fig.~2.
    \item \textbf{Varying the Mismatch Level $\mmf$:} The maximum inter-execution step is fixed to $M = 5$ while the mismatch level is varied as $\mmf = 0.001i$ for $i \in \bI_{1:10}$ under both scenarios. The simulation results are shown in Fig.~3.
\end{enumerate}
\begin{figure}
    \centering
    \includegraphics[width=\linewidth]{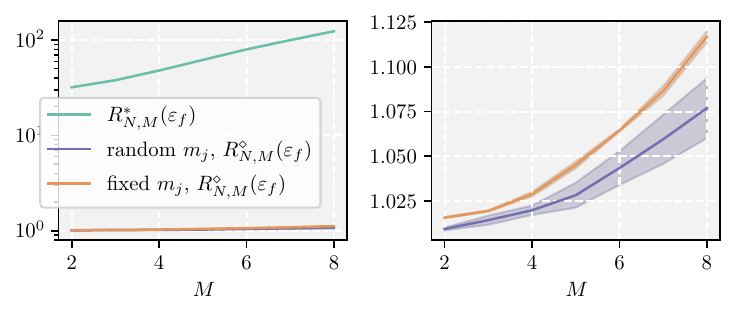}
    \vspace{-0.7cm}
    \caption{Performance evaluation for varying maximum inter-execution intervals $M$. The shaded envelope encloses closed-loop results across four initial states and/or 100 random realizations of $m_j$. The right panel presents a zoomed-in view of the left panel.}
    \label{fig:varying_inter_execution}
\end{figure}
\begin{figure}
    \centering
    \includegraphics[width=\linewidth]{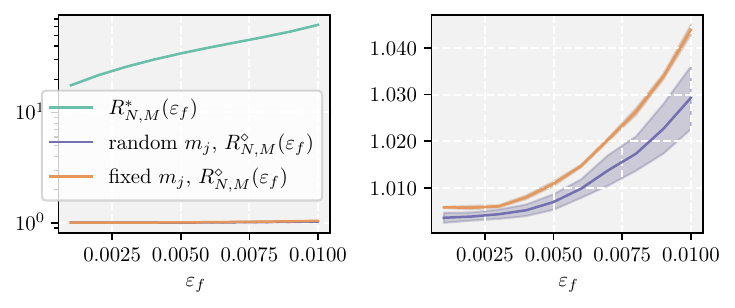}
    \vspace{-0.7cm}
    \caption{Performance evaluation under varying mismatch level $\mmf$. The shaded envelope encloses closed-loop results across four initial states and/or 100 random realizations of $m_j$. The right panel presents a zoomed-in view of the left panel.}
    \label{fig:varying_mismatch}
\end{figure}
As shown in Figs.~2 and~3, the competitive-ratio bound consistently satisfies $\Rcr \geq R^{\diamond}_{N,M}(\mmf)$, corroborating the theoretical guarantee. Although conservative, as is typical of worst-case theoretical performance bounds~\cite{kohler2023stability, grune2017nonlinear, shi2025suboptimality}, it captures the observed performance deterioration with increasing $M$ and $\mmf$. Its significance lies not merely in reproducing this trend, but in providing a theoretically certified, parameter-dependent bound rather than an empirical fit, thereby supporting performance-aware choices of inter-execution intervals and model accuracy.

%% file: 1-major_content/6-conclusion.tex
\section{Conclusions and future work}
\label{sec:6-conclusion} 
This paper has established a theoretical analysis pipeline for certainty-equivalence model predictive control (CE-MPC) subject to model mismatch and asynchronous feedback. By leveraging perturbation analysis of the MPC value function, we have derived explicit guarantees for closed-loop stability and competitive-ratio performance under general terminal costs. Ultimately, these results systematically characterize how the prediction horizon, feedback frequency, and model mismatch jointly influence the stability and worst-case infinite-horizon performance of asynchronous CE-MPC. 

Promising directions for future research include deriving tighter perturbation bounds and extending this framework to multi-agent systems with distributed MPC.

%% file: 1-major_content/7-appendices.tex
\appendices
\setlength{\abovedisplayskip}{4pt plus 1pt minus 3pt}
\setlength{\belowdisplayskip}{4pt plus 1pt minus 3pt}
\setlength{\abovedisplayshortskip}{0pt plus 2pt}
\setlength{\belowdisplayshortskip}{3pt plus 1pt minus 1pt}
{\small
\input{5-appendices/1-appendix_A}

\input{5-appendices/2-appendix_B}

}

%% file: 5-appendices/1-appendix_A.tex
\section{Details of Section \ref{sec:4-theoretical_analysis}}
\label{appendix:A}
Given an initial state $x \in \cX$ and a state-feedback policy $\pi:\cX \to \cU$, the $k$-step-forward closed-loop state of the nominal model $f$ under $\pi$ is denoted by $\phi_{\pi,k}(x)$, satisfying $\phi_{\pi,0}(x) = x$ and $\phi_{\pi,k+1}(x) = f(\phi_{\pi,k}(x), \pi(\phi_{\pi,k}(x)))$. In addition, for a given policy $\pi$, we introduce the shorthand notation $\ell_{\pi}(x) := \ell(x, \pi(x))$.
\vspace{-0.3cm}
\subsection{Auxiliary Lemmas}
\label{appendix:A-lemmas}
\begin{lemma}
    \label{lm:terminal_policy}
    Let Assumption~\ref{ass:terminal_cost_relaxed_clf} hold. Given a state $x \in \xroa(\ell)$, there exists a policy $\pit: \cX \to \cU$ such that $\forall k \in \bI_{1:\infty}$,
    \begin{equation}
        \label{eq:appA-multistep}
        \elf(\phi_{\pit,k}(x)) \splus \sum^{k-1}_{i=0}\ell_{\pit}(\phi_{\pit,i}(x)) \leq (1 \splus \epf)^{k}\elf(x),
    \end{equation}
    where the constant $\epf$ is given in~\eqref{eq:relaxed_clf} in Assumption~\ref{ass:terminal_cost_relaxed_clf}.
\end{lemma}
\begin{proof}
    Denote the minimizer satisfying~\eqref{eq:relaxed_clf} by $u^\star_{\mathrm{T}}(x)$, which implicitly defines a terminal policy $\pit: \cX \to \cU$ as $\pit(x) = u^\star_{\mathrm{T}}(x)$. Under the policy $\pit$, inequality \eqref{eq:relaxed_clf} implies that
    \begin{equation}
        \label{eq:appA-terminal_policy_substitution}
        \elf(f(x,\pit(x))) + \ell_{\pit}(x) \leq (1+\epf)\ell^\star(x).
    \end{equation}
    Given $k \in \bI_{1:\infty}$, for all $i \in \bI_{0:k}$, it holds that
    \begin{multline}
        \label{eq:appA-recursion_terminal_policy}
        (1\splus\epf)^{k-i}\elf(\phi_{\pit,i}(x)) +  \ell_{\pit}(\phi_{\pit,i-1}(x)) \\ \leq(1\splus\epf)^{k-i}\left[\elf(\phi_{\pit,i}(x))+\ell_{\pit}(\phi_{\pit,i-1}(x))\right] \\
        \overset{\eqref{eq:appA-terminal_policy_substitution};\;\ell^\star(x)\leq \elf(x)}{\leq} (1\splus\epf)^{k-i+1}\elf(\phi_{\pit,i-1}(x)).
    \end{multline}
    Starting with $\elf(\phi_{\pit,k}(x)) \splus \sum^{k-1}_{i=0}\ell_{\pit}(\phi_{\pit,i}(x))$, recursively applying~\eqref{eq:appA-recursion_terminal_policy} for $i = k,k-1,\dots,1$ yields~\eqref{eq:appA-multistep}.
\end{proof}
\begin{lemma}
    \label{lm:increasing_function}
    Given positive constants $c_1$, $c_2$, and $\bar{y}$, define the function $g:\mathbb{R}_+ \to \mathbb{R}_+$ as $g(y) = [g_1(y)][g_2(y)]^{-1}$, where $g_1(y) = (1+c_1)^y-1$ and $g_2(y) = c_2 + (\bar{y} - y)(1+c_1)^y$. Then, $g$ is strictly monotonically increasing on $(0, \bar{y})$.
\end{lemma}
\begin{proof}
    Let $h(y) = (1+c_1)^y$. By the quotient rule, the derivative of $g$ is given by $g'(y) = [g_2(y)]^{-2} h(y) [c_2 \ln(1+c_1) + h(y) - 1 + (\bar{y} - y)\ln(1+c_1)]$. For any $y \in (0, \bar{y})$, it holds that $h(y) > 1$, $\bar{y} - y > 0$, and $\ln(1+c_1) > 0$, leading to $g'(y) > 0$. Thus, $g$ is strictly monotonically increasing on $(0, \bar{y})$. 
\end{proof}
\vspace{-0.5cm}
\subsection{Proof of Theorem~\ref{thm:nominal_stability_performance}}
\label{appendix:A-proof-thm2}
\begin{proof}
By the optimality of $V_N(x_{t_{j+1}};f)$, it holds that
\begin{multline}
    \label{eq:appA-value_propagation_suboptimal}
    V_N(x_{t_{j+1}};f) \leq \sum^{N-1}_{k=m_j}\ell_{k|t_j}(f) + \sum^{m_j-1}_{k=0}\ell_{\pit}(\phi_{\pit,k}(\xi^\star_{N|t_j}(f))) \\ + \elf(\phi_{\pit,m_j}(\xi^\star_{N|t_j}(f)) \\
    \overset{\eqref{eq:appA-multistep}}{\leq} \sum^{N-1}_{k=m_j}\ell_{k|t_j}(f) + (1\splus\epf)^{m_j}\elf(\xi^\star_{N|t_j}(f)).
\end{multline}
Therefore, the dissipation term $V_N(x_{t_{j+1}};f) - V_N(x_{t_{j}};f)$ satisfies
\begin{multline}
    \label{eq:appA-dissipation_original}
    V_N(x_{t_{j+1}};f) - V_N(x_{t_{j}};f) \leq -\sum^{m_j-1}_{k=0}\ell_{k|t_j}(f) \\ + \left[(1+\epf)^{m_j}-1\right]\elf(\xi^\star_{N|t_j}(f)).
\end{multline}
The remaining proof has two parts with complementary conditions. First, assume that the following condition holds:
\begin{equation}
    \label{eq:appA-case_distinction}
    \mathcal{E}^+: \elf(\xi^\star_{N|t_j}(f)) \leq \Gamma_{N,m_j}\sum^{m_j-1}_{k=0}\ell^\star(\xi^\star_{k|t_j}(f)),
\end{equation}
where $\Gamma_{N,m_j} := \frac{(1+\umpc)\mpc_N}{(1+\umpc)+(N-m_j)(1+\epf)^{m_j}}$. Then \eqref{eq:appA-dissipation_original} yields
\begin{multline}
    \label{eq:appA-distinction_original_final}
    V_N(x_{t_{j+1}};f) - V_N(x_{t_{j}};f) \overset{\eqref{eq:appA-case_distinction}}{\leq} -\sum^{m_j-1}_{k=0}\ell_{k|t_j}(f) \\ + \left[(1+\epf)^{m_j}-1\right]\Gamma_{N,m_j}\sum^{m_j-1}_{k=0}\ell^\star(\xi^\star_{N|t_j}(f))
    \\
    \leq -\left(1 - \frac{(1+\umpc)\mpc_N[(1 + \epf)^{m_j}-1]}{1+\umpc+(N-m_j)(1+\epf)^{m_j}}\right)\sum^{m_j-1}_{k=0}\ell_{k|t_j}(f), \hspace{-0.6em}
\end{multline}
where the second inequality is due to $\ell^\star(\xi^\star_{N|t_j}(f)) \leq \ell_{k|t_j}(f)$. On the other hand, assume that the following condition holds:
\begin{equation}
    \label{eq:appA-case_distinct_reverse}
    \mathcal{E}^-: \elf(\xi^\star_{N|t_j}(f)) \geq \Gamma_{N,m_j}\sum^{m_j-1}_{k=0}\ell^\star(\xi^\star_{k|t_j}(f)).
\end{equation}
Then, by the definition of $V_N(x;f)$, it holds that
\begin{multline}
    \label{eq:appA-disctinction_reverse_mpcvalue_starting}
    \sum^{N-1}_{k=0}\ell^\star(\xi^\star_{k|t_j}(f)) \leq V_N(x_{t_j};f) - \elf(\xi^\star_{N|t_j}(f)) \\[-0.5em]
    \overset{\eqref{eq:cost_controllability};\;\eqref{eq:appA-case_distinct_reverse}}{\leq} [1 + \mpc_N - \Gamma_{N,m_j}]\sum^{m_j-1}_{k=0}\ell^\star(\xi^\star_{k|t_j}(f)).
\end{multline}
Subtracting $\sum^{m_j-1}_{k=0}\ell^\star(\xi^\star_{k|t_j}(f))$ from both sides of~\eqref{eq:appA-disctinction_reverse_mpcvalue_starting} yields
\begin{equation}
    \label{eq:appA-disctinction_reverse_mpcvalue_subtracted}
    \sum^{N-1}_{k=m_j}\ell^\star(\xi^\star_{k|t_j}(f)) \leq [\mpc_N - \Gamma_{N,m_j}]\sum^{m_j-1}_{k=0}\ell^\star(\xi^\star_{k|t_j}(f)),
\end{equation}
which implies that $\exists k_{0} \in \bI_{m_j:N-1}$ such that
\begin{equation}
    \label{eq:appA-disctinction_reverse_mpcvalue_existence}
    \ell^\star(\xi^\star_{k_0|t_j}(f)) \leq \left[\frac{\mpc_N - \Gamma_{N,m_j}}{N-m_j}\right]\sum^{m_j-1}_{k=0}\ell^\star(\xi^\star_{k|t_j}(f)).
\end{equation}
Again, by the optimality of $V_N(x_{t_{j+1}};f)$, it holds that
\begin{align}
    \label{eq:appA-distinction_reverse_final}
    & V_N(x_{t_{j+1}};f) + \sum^{m_j-1}_{k=0}\ell_{k|t_j}(f) - V_N(x_{t_{j}};f) \notag\\[-0.5em] 
    & \hspace{2em} \leq V_{N-k_0+m_j}(\xi^\star_{k_0|t_j}(f)) + \sum^{k_0-1}_{k=0}\ell_{k|t_j}(f) - V_N(x_{t_{j}};f) \notag \\[-0.5em]
    & \overset{\eqref{eq:cost_controllability};\;\eqref{eq:appA-case_distinct_reverse};\;\eqref{eq:appA-disctinction_reverse_mpcvalue_existence}}{\leq} \left\{\frac{(1+\umpc)[\mpc_N - \Gamma_{N,m_j}]}{N-m_j} - \Gamma_{N,m_j}\right\}\sum^{m_j-1}_{k=0}\ell_{k|t_j}(f) \notag \\[-0.5em]
    & \hspace{2em} \leq \frac{(1+\umpc)\mpc_N[(1 + \epf)^{m_j}-1]}{1+\umpc+(N-m_j)(1+\epf)^{m_j}}\sum^{m_j-1}_{k=0}\ell_{k|t_j}(f) 
\end{align}
By the law of trichotomy, the condition $\mathcal{E}^+ \lor \mathcal{E}^-$ is exhaustively complete. Consequently, leveraging~\eqref{eq:appA-distinction_original_final} and~\eqref{eq:appA-distinction_reverse_final} yields
\begin{equation}
    \label{eq:appA-dissipation_final}
    V_N(x_{t_{j+1}};f) - V_N(x_{t_{j}};f) \leq -[1 - \delta_{N,m_j}] \sum^{m_j-1}_{k=0}\ell_{k|t_j}(f),
\end{equation}
where $\delta_{N,m_j} = \frac{(1+\umpc)\mpc_N[(1 + \epf)^{m_j}-1]}{1+\umpc+(N-m_j)(1+\epf)^{m_j}}$. By Lemma~\ref{lm:increasing_function}, it holds that $\delta_{N,m_j} \leq \delta_{N,M}$ with $\delta_{N,M}$ given in~\eqref{eq:nominal_stability_margin}. Thus, \eqref{eq:appA-dissipation_final} reduces to
\begin{equation}
    \label{eq:appA-dissipation_final_general}
    V_N(x_{t_{j+1}};f) - V_N(x_{t_{j}};f) \leq -[1 - \delta_{N,M}] \sum^{m_j-1}_{k=0}\ell_{k|t_j}(f).
\end{equation}
Applying~\eqref{eq:appA-dissipation_final_general} inductively from $x_0 = x_{1,0} = x$ via a telescopic sum~\cite{grune2017nonlinear, kohler2023stability} leads to
\begin{equation}
    \label{eq:appA-telescopic}
    [1 - \delta_{N,M}]\sum^{J}_{j=1}\sum^{m_j-1}_{k=0}\ell_{k|t_j}(f) \leq V_N(x;f) - V_N(x_{t_{J+1}};f)
\end{equation}
Taking the limit as $J \to \infty$ and noting that $V_N(x_{t_{J+1}};f) \geq 0$, \eqref{eq:appA-telescopic} establishes that
\begin{equation}
    \label{eq:appA-telescopic_limit}
    [1 - \delta_{N,M}]\sum^{\infty}_{j=1}\sum^{m_j-1}_{k=0}\ell_{k|t_j}(f) \leq V_N(x;f).
\end{equation}
Following standard arguments in MPC stability proofs using relaxed dynamic programming~\cite{grune2008infinite, grune2017nonlinear}, since $V_N(x;f) < +\infty$, the condition $1 - \delta_{N,M} > 0$ implies that $\lim_{j\to\infty}\ell_{k|t_j}(f) = 0$. By the positive definiteness of $\ell$, this further ensures that $\lim_{j\to\infty}x_{j,k} = \bzero$. Consequently, the origin is asymptotically stable under the asynchronous MPC controller without model mismatch provided that $\delta_N(M) < 1$. Furthermore, analogous to~\cite[Proof of Theorem~5, Part~II]{kohler2023stability}, it holds that
\begin{equation}
    \label{eq:appA-mpcvalue_to_infty}
    V_N(x;f) \leq 1 + \frac{\eigmax{P}}{\eigmin{Q}}\left(\frac{\umpc}{1+\umpc}\right)^N V_{\infty}(x;f).
\end{equation}
Combining~\eqref{eq:appA-telescopic_limit} and~\eqref{eq:appA-mpcvalue_to_infty} yields
\begin{equation}
    \sum^{\infty}_{j=1}\sum^{m_j-1}_{k=0}\ell_{k|t_j}(f) \leq \frac{1 + \frac{\eigmax{P}}{\eigmin{Q}}\left(\frac{\umpc}{1+\umpc}\right)^N}{1 - \delta_{N,M}}V_{\infty}(x;f),
\end{equation}
which confirms the performance bound~\eqref{eq:nominal_performance_bound}.
\end{proof}

%% file: 5-appendices/2-appendix_B.tex
\section{Details of Section \ref{sec:theoreticalAnalysis_model_mismatch}}
\label{appendix:B}
\input{5-appendices/2.1-lemmas}
\input{5-appendices/2.2-propositions}
\input{5-appendices/2.3-final_theorem}

%% file: 5-appendices/2.1-lemmas.tex
\subsection{Auxiliary Lemmas}
\label{appendix:B-1-lemmas}
\begin{proof}[Proof of Lemma~\ref{lem:state_perturbation}]
    The proof proceeds by induction on $k$. First, for the base case $k=0$, the relation $\|\psi^\ast_{1|t}(\bv_{0:n}) - \psi_{1|t}(\bv_{0:n})\| = [L^\ast_f(\mmf)]^0\|\Df_t(\psi_{0|t}(\bv_{0:n}), v_0)\|$ implies that \eqref{eq:state_perturbation} holds trivially. Next, assume the induction hypothesis holds for $k = k_0$, i.e.,
    \begin{multline}
        \label{eq:appB.1-induction_start}
        \hspace{-1em}\|\psi^\ast_{k_0+1|t}(\bv_{0:n}) - \psi_{k_0+1|t}(\bv_{0:n})\| \\ 
        \hspace{-0.6em}\leq \sum^{k_0}_{i=0} \left[L^\ast_f(\mmf)\right]^{k_0-i} \|\Df_{t+i}(\psi_{i|t}(\bv_{0:n}), v_i)\|. \hspace{-0.4em}
    \end{multline}
    Then, for the inductive step $k = k_0+1$, it follows that
    \begin{align}
        \label{eq:appB.1-induction_derivation}
        & \|\psi^\ast_{k_0+2|t}(\bv_{0:n}) - \psi_{k_0+2|t}(\bv_{0:n})\| \notag \\
        \leq \;& \|f(\psi^\ast_{k_0+1|t}(\bv_{0:n}), v_{k_0+1}) - f(\psi_{k_0+1|t}(\bv_{0:n}), v_{k_0+1})\| \notag \\
        & \hspace{8em} + \|\Df_{t+k_0+1}(\psi^\ast_{k_0+1|t}(\bv_{0:n}), v_{k_0+1})\| \notag \\[-0.2cm]
        & \hspace{-2.2em}\overset{\eqref{eq:Lipschitz_nominal};\;\eqref{eq:Lipschitz_uncertainty}}{\leq} (L_{f} + \bar{L}_{\Df})\|\psi^\ast_{k_0+1|t}(\bv_{0:n}) - \psi_{k_0+1|t}(\bv_{0:n})\| \notag \\
        & \hspace{7em} + \|\Df_{t+k_0+1}(\psi_{k_0+1|t}(\bv_{0:n}), v_{k_0+1})\| \notag \\[-0.7em]
        \overset{\eqref{eq:appB.1-induction_start}}{\leq} 
        & \sum^{k_0+1}_{i=0} (L_{f} + \bar{L}_{\Df})^{k_0+1-i} \|\Df_{t+i}(\psi_{i|t}(\bv_{0:n}), v_i)\|,
    \end{align}
    which confirms that \eqref{eq:state_perturbation} holds for $k = k_0+1$. The upper bound can also be established in terms of $\Df_{t+i}(\psi_{i|t}(\bv_{0:n}), v_i)$, and its proof also adopts induction. For the base case $k=0$, $\|\psi^\ast_{1|t}(\bv_{0:n}) - \psi_{1|t}(\bv_{0:n})\| = (L_f)^0\|\Df_t(\psi^\ast_{0|t}(\bv_{0:n}), v_0)\|$. Then, assume the induction hypothesis holds for $k=k_0$, i.e.,
    \begin{multline}
        \label{eq:appB.1-induction_start_2}
        \hspace{-1em}\|\psi^\ast_{k_0+1|t}(\bv_{0:n}) - \psi_{k_0+1|t}(\bv_{0:n})\| \\ 
        \hspace{-0.6em}\leq \sum^{k_0}_{i=0} (L_f)^{k_0-i} \|\Df_{t+i}(\psi^\ast_{i|t}(\bv_{0:n}), v_i)\|. \hspace{-0.4em}
    \end{multline}
    Then, for the inductive step $k = k_0+1$, it follows that
    \begin{align}
        \label{eq:appB.1-induction_derivation_2}
        & \|\psi^\ast_{k_0+2|t}(\bv_{0:n}) - \psi_{k_0+2|t}(\bv_{0:n})\| \notag \\
        \leq \;& \|f(\psi^\ast_{k_0+1|t}(\bv_{0:n}), v_{k_0+1}) - f(\psi_{k_0+1|t}(\bv_{0:n}), v_{k_0+1})\| \notag \\
        & \hspace{8em} + \|\Df_{t+k_0+1}(\psi^\ast_{k_0+1|t}(\bv_{0:n}), v_{k_0+1})\| \notag \\[-0.2cm]
        & \hspace{-1.5em} \overset{\eqref{eq:Lipschitz_nominal}}{\leq}  L_f\|\psi^\ast_{k_0+1|t}(\bv_{0:n}) - \psi_{k_0+1|t}(\bv_{0:n})\| \notag \\
        & \hspace{8em} + \|\Df_{t+k_0+1}(\psi^\ast_{k_0+1|t}(\bv_{0:n}), v_{k_0+1})\| \notag \\[-0.7em]
        \overset{\eqref{eq:appB.1-induction_start_2}}{\leq}
        & \sum^{k_0+1}_{i=0} (L_f)^{k_0+1-i} \|\Df_{t+i}(\psi_{i|t}(\bv_{0:n}), v_i)\|,
    \end{align}
    which completes induction transition and confirms that
    \begin{multline} \label{eq:state_perturbation_2}
    \hspace{-1em}\|\psi^\ast_{k+1|t}(\bv_{0:n}) - \psi_{k+1|t}(\bv_{0:n})\| \\ 
    \hspace{-0.6em}\leq \sum^{k}_{i=0} (L_f)^{k-i} \|\Df_{t+i}(\psi^\ast_{i|t}(\bv_{0:n}), v_i)\|, \forall k \in \bI_{0:n}. \hspace{-0.4em}
    \end{multline}
    Note that both~\eqref{eq:state_perturbation} and~\eqref{eq:state_perturbation_2} are valid, and their upper bounds share a symmetric structure.
\end{proof}
\begin{proof}[Proof of Lemma~\ref{lem:error_matching}]
    The proof directly follows that of~\cite[Proposition~1]{liu2026certainty} by noting that quadratic functions of the form $\ell_x(x) = x^\top Qx$ with $Q \succ 0$ are $\eigmin{Q}$-strongly convex.
\end{proof}
\begin{lemma}
    \label{lm:state_mismatch_scaled}
    Let Assumption~\ref{ass:system_setting} hold. For any time step $t$ and $k \in \bI_{1:N}$, the state perturbation $\|\eta^\ast_{k|t}(f) - \xi^\star_{k|t}(f)\|$ satisfies
    \begin{subequations}
        \label{eq:appB.2-cost_scaled_state_perturbation}
        \begin{align}
        \label{eq:appB.2-cost_scaled_state_perturbation_1}
        &\|\eta^\ast_{k|t}(f) - \xi^\star_{k|t}(f)\|^2 \leq \bar{\zeta}_k[L_{\Df}(\mmf)]^2\sum^{k-1}_{i=0}\ell_{i|t}(f),\\[-0.3em]
        \label{eq:appB.2-cost_scaled_state_perturbation_2}
        &\|\eta^\ast_{k|t}(f) - \xi^\star_{k|t}(f)\|^2 \leq \zeta_k[L_{\Df}(\mmf)]^2\sum^{k-1}_{i=0}\ell^\ast_{i|t}(f),
        \end{align}
    \end{subequations}
    where $\bar{\zeta}_k = \mu_{\mathrm{ce}}\sum^{k-1}_{i=0}(L_f+\bar{L}_{\Df})^{2i}$ and $\zeta_k = \mu_{\mathrm{ce}}\sum^{k-1}_{i=0}(L_f)^{2i}$ given in Lemma~\ref{lem:error_matching}. 
\end{lemma}
\begin{proof}
    The proof directly applies Lemmas~\ref{lem:state_perturbation} and~\ref{lem:error_matching} as follows:
    \begin{align}
    \label{eq:appB.2-proof_lem5}
        & \|\eta^\ast_{k|t}(f) - \xi^\star_{k|t}(f)\|^2 \notag \\
        \overset{\eqref{eq:state_perturbation}}{\leq} \;&\left(\sum^{k-1}_{i=0}(L_{f} + \bar{L}_{\Df})^{k-1-i} \|\Df_{t+i}(\xi^\star_{k|t}(f), \nu^\star_{k|t}(f))\|\right)^2 \notag \\
        \leq\;\;& \sum^{k-1}_{i=0}(L_f+\bar{L}_{\Df})^{2i}\sum^{k-1}_{i=0}\|\Df_{t+i}(\xi^\star_{k|t}(f), \nu^\star_{k|t}(f))\|^2 \notag \\
        \overset{\eqref{eq:mismatch_bound}}{\leq}\; & \mu_{\mathrm{ce}}\left[\sum^{k-1}_{i=0}(L_f+\bar{L}_{\Df})^{2i}\right][L_{\Df}(\mmf)]^2\sum^{k-1}_{i=0}\ell_{i|t}(f),
    \end{align}
    which confirms~\eqref{eq:appB.2-cost_scaled_state_perturbation_1}. Likewise, \eqref{eq:appB.2-cost_scaled_state_perturbation_2} can be established following the same proof as in~\eqref{eq:appB.2-proof_lem5} by leveraging~\eqref{eq:state_perturbation_2} and~\eqref{eq:mismatch_bound}.
\end{proof}

\begin{lemma} 
\label{lm:quadratic_error_expansion}
Given vectors $x,y \in \bR^n$ and a positive definite matrix $Q \in \bR^{n\times n}$ ($Q \succ 0$), it holds that
\begin{equation}
    \label{eq:appB.2-pre_quadratic_lemma}
    |\nQ{x} - \nQ{y}| \leq \eigmax{Q}(\|x- y\|^2 + 2\|x- y\|\min\{\|x\|,\|y\|\}).
\end{equation}
\end{lemma}
\begin{proof}
By the Cauchy-Schwarz inequality,
    \begin{multline}
    \label{eq:appB.2-pre_quadratic_lemma_proof_1}
    |\nQ{x} - \nQ{y}| = |\nQ{(x-y)+y} - \nQ{y}| \leq \nQ{x-y} \\ + 2|(x-y)Q^\top y| \leq \eigmax{Q}(\|x-y\|^2 + 2\|x-y\|\|y\|)
\end{multline}
Similarly, it follows that
\begin{multline}
    \label{eq:appB.2-pre_quadratic_lemma_proof_2}
    |\nQ{x} - \nQ{y}| = |\nQ{x} - \nQ{(y-x)+x}| \leq \nQ{x-y} \\ + 2|(y-x)Q^\top x| \leq \eigmax{Q}(\|x-y\|^2 + 2\|x-y\|\|x\|)
\end{multline}
Combining~\eqref{eq:appB.2-pre_quadratic_lemma_proof_1} and~\eqref{eq:appB.2-pre_quadratic_lemma_proof_2} yields~\eqref{eq:appB.2-pre_quadratic_lemma}.
\end{proof}

%% file: 5-appendices/2.2-propositions.tex
\subsection{Supporting Propositions}
\label{appendix:B-2-propositions}
\input{5-appendices/2.2.1-pp1}
\input{5-appendices/2.2.2-pp2}
\input{5-appendices/2.2.3-pp3}

%% file: 5-appendices/2.2.1-pp1.tex
\begin{proof}[Proof of Proposition~\ref{prop:state_perturbation}]
    To study state-driven perturbation, several variables are explicitly conditioned on $x_t = x$ by expanding the notations $\xi^\star_{k|t}(f)$, $\nu^\star_{k|t}(f)$, $\ell_{k|t}(f)$, $\ell_{k|t}(f^\ast)$, $\psi_{k|t}(\bv_{0:n})$, and $\psi^\ast_{k|t}(\bv_{0:n})$ to $\xi^\star_{k|t}(x;f)$, $\nu^\star_{k|t}(x;f)$, $\ell_{k|t}(x;f)$, $\ell_{k|t}(x;f^\ast)$, $\psi_{k|t}(x,\bv_{0:n})$, and $\psi^\ast_{k|t}(x,\bv_{0:n})$, respectively. For brevity, $\psi_{k|t}(x',\{\nu^\star_{i|t}(x;f)\}^{N-1}_{i=0})$ is denoted as $\psi_{k|t}(x';\bm{\nu}^\star(x,f))$ hereafter. For Lipschitz-continuous models, the following perturbation bound is a well-established result in the literature~\cite{wabersich2022cautious, liu2026certainty}:
\begin{equation}
    \label{eq:appB.2-trivial_Lipschitz_perturbation_bound}
    \|\psi_{k|t}(x';\bm{\nu}^\star(x,f)) - \xi^\star_{k|t}(x;f)\| \leq (L_f)^k\|x' - x\|,
\end{equation}
which quantifies the state trajectory deviation resulting from differing initial states under identical input sequences, and can be readily verified via induction. An upper bound on the general value function perturbation $V_N(x';f) - V_N(x;f)$ is first established. By the optimality of $V_N(x';f)$, it holds that
\begin{multline}
    \label{eq:appB.2-beginning_optimality_bound}
    V_N(x';f) \leq \sum^{N-1}_{k=0}\ell(\psi_{k|t}(x';\bm{\nu}^\star(x,f)), \nu^\star_{k|t}(x;f)) \\[-0.2cm] + \elf(\psi_{N|t}(x';\bm{\nu}^\star(x,f))).
\end{multline}
Consequently, $V_N(x';f) - V_N(x;f)$ can be bounded as
\begin{multline}
    \label{eq:appB.2-proceed_input_cancellation}
    V_N(x';f) - V_N(x;f) \overset{\eqref{eq:appB.2-beginning_optimality_bound}}{\leq} \sum^{N-1}_{k=0}\nQ{\psi_{k|t}(x';\bm{\nu}^\star(x,f))}\\
    + \nP{\psi_{N|t}(x';\bm{\nu}^\star(x,f))} - \sum^{N-1}_{k=0}\nQ{\xi^\star_{k|t}(x;f)}
    - \nP{\xi^\star_{N|t}(x;f)} \\
    \overset{\eqref{eq:appB.2-pre_quadratic_lemma};\;\eqref{eq:appB.2-trivial_Lipschitz_perturbation_bound}}{\leq}
    \underbrace{\eigmax{Q}\sum^{N-1}_{k=0}(L_f)^{2k} + \eigmax{P}(L_f)^{2N}}_{\coloneqq \sigma_{N,2}}\|x'-x\|^2 \\[-0.5em]
    + 2\eigmax{Q}\sum^{N-1}_{k=0}(L_f)^k\|\xi^\star_{k|t}(x;f)\|\|x'-x\| \\
    + 2\eigmax{P}(L_f)^N\|\xi^\star_{N|t}(x;f)\|\|x'-x\|.
\end{multline}
Applying the Cauchy-Schwarz inequality yields
\begin{align}
    \label{eq:appB.2-proceed_first_order_coefficient}
    & \;\eigmax{Q}\sum^{N-1}_{k=0}(L_f)^k\|\xi^\star_{k|t}(x;f)\| + \eigmax{P}(L_f)^N\|\xi^\star_{N|t}(x;f)\| \notag \\
    \leq & \Bigg\{\left(\sum^{N-1}_{k=0}\nQ{\xi^\star_{k|t}(x;f)}+\nP{\xi^\star_{N|t}(x;f)}\right) \notag \\[-0.5em]
    & \hspace{10em}\underbrace{\Bigg[\frac{\eigmax{Q}^2}{\eigmin{Q}}\sum^{N-1}_{k=0}(L_f)^{2k} + \frac{\eigmax{P}^2}{\eigmin{P}}(L_f)^{2N}\Bigg]}_{\coloneqq \sigma_{N,1}}\Bigg\}^{\frac{1}{2}} \notag\\
    \leq & \left\{\sigma_{N,1}V_N(x;f)\right\}^{\frac{1}{2}}.
\end{align}
Finally, leveraging~\eqref{eq:appB.2-proceed_input_cancellation} and~\eqref{eq:appB.2-proceed_first_order_coefficient} leads to the general bound:
\begin{multline}
    \label{eq:appB.2-general_state_perturbation}
    V_N(x';f) - V_N(x;f) \leq \sigma_{N,2}\|x'-x\|^2 \\ + 2\left\{\sigma_{N,1}V_N(x;f)\right\}^{\frac{1}{2}}\|x'-x\|.
\end{multline}
Applying~\eqref{eq:appB.2-general_state_perturbation} with $x = \xi^\star_{k|t}(f)$ and $x'= \eta_{k|t}(f^\ast)$ yields
\begin{align}
    \label{eq:appB.2-final_derviation_state_perturbation_bound}
    & V_N(\eta_{k|t}(f^\ast);f) - V_N(\xi^\star_{k|t}(f);f) \leq \sigma_{N,2}\|\eta_{k|t}(f^\ast) - \xi^\star_{k|t}(f)\|^2 \notag \\
    & \hspace{5em} + 2\left\{\sigma_{N,1}V_N(\xi^\star_{k|t}(f);f)\right\}^{\frac{1}{2}}\|\eta_{k|t}(f^\ast) - \xi^\star_{k|t}(f)\| \notag\\[-0.2cm]
    & \overset{\eqref{eq:cost_controllability};\;\eqref{eq:appA-dissipation_final}}{\leq} \sigma_{N,2}\|\eta_{k|t}(f^\ast) - \xi^\star_{k|t}(f)\|^2 \notag \\[-0.2cm]
    & \hspace{1em} + 2\left\{\sigma_{N,1}(\mpc_N + \delta_{N,k})\sum^{k-1}_{i=0}\ell_{i|t}(f)\right\}^{\frac{1}{2}}\|\eta_{k|t}(f^\ast) - \xi^\star_{k|t}(f)\| \notag\\
    & \hspace{0.5em} \overset{\eqref{eq:appB.2-cost_scaled_state_perturbation_1}}{\leq} \underbrace{\sigma_{N,2}\bar{\zeta}_k}_{\coloneqq \alpha^{(2)}_{N,k}}[L_{\Df}(\mmf)]^2\sum^{k-1}_{i=0}\ell_{i|t}(f) \notag\\[-0.6cm]
    & \hspace{6em} +\underbrace{2\left\{\sigma_{N,1}(\mpc_N + \delta_{N,k})\bar{\zeta}_k\right\}^{\frac{1}{2}}}_{\coloneqq \alpha^{(1)}_{N,k}}L_{\Df}(\mmf)\sum^{k-1}_{i=0}\ell_{i|t}(f)\notag \\[-0.2cm]
    & \hspace{1em}\leq \left\{\alpha^{(1)}_{N,k}L_{\Df}(\mmf) + \alpha^{(2)}_{N,k}[L_{\Df}(\mmf)]^2 \right\}\sum^{k-1}_{i=0}\ell_{i|t}(f),
\end{align}
where $\delta_{N,k} \seq \frac{(1+\umpc)\mpc_N[(1 + \epf)^{k}-1]}{1+\umpc+(N-k)(1+\epf)^{k}}$ and $\bar{\zeta}_k$ is given in~\eqref{eq:appB.2-cost_scaled_state_perturbation_1} in Lemma~\ref{lm:state_mismatch_scaled}, confirming the state-driven perturbation bound~\eqref{eq:state_perturbation_original}.
\end{proof}

%% file: 5-appendices/2.2.2-pp2.tex
\begin{proof}[Proof of Proposition~\ref{prop:model_mismatch}]
    Recall the definition of $\xi^\star_{k|t}(f)$ and $\nu^\star_{k|t}(f)$ in Section~\ref{subsec2.3-ampc}. Analogously, the optimal solution sequences to the problem $\probmpc{x_t}{\tsysn}$ are denoted by $\{\xi^\star_{k|t}(f^\ast)\}^N_{k=0}$ and $\{\nu^\star_{k|t}(f^\ast)\}^{N-1}_{k=0}$. Consistent with the definition of $\eta^\ast_{k|t}(f)$ in Section~\ref{sec4.5-1:perturbation_analysis}, define $\eta_{k|t}(f^\ast) := \psi_{k|t}(\{\nu^\star_{k|t}(f^\ast)\}^{N-1}_{k=0})$ for $k \in \bI_{0:N}$, which represents the state evolution under the nominal model $f$ when driven by the optimal control inputs $\{\nu^\star_{k|t}(f^\ast)\}^{N-1}_{k=0}$. The corresponding stage costs are written compactly as $\ell^\ast_{k|t}(f^\ast) \coloneqq \ell(\xi^\star_{k|t}(f^\ast), \nu^\star_{k|t}(f^\ast))$ and $\ell_{k|t}(f^\ast) = \ell(\eta_{k|t}(f^\ast), \nu^\star_{k|t}(f^\ast))$.
    
    Similar to~\eqref{eq:appB.2-cost_scaled_state_perturbation_2}, and following the same reasoning as in~\eqref{eq:appB.2-proof_lem5}, the following perturbation bound can be readily established:
\begin{equation}
    \label{eq:appB.2-scaled_state_perturbation_additional}
    \|\eta_{k|t}(f^\ast) - \xi^\star_{k|t}(f^\ast)\|^2 \leq \zeta_k[L_{\Df}(\mmf)]^2\sum^{k-1}_{i=0}\ell^\ast_{i|t}(f^\ast),
\end{equation}
where $\zeta_k$ is defined in~\eqref{eq:appB.2-cost_scaled_state_perturbation_2} of Lemma~\ref{lm:state_mismatch_scaled}. By the optimality of $V_N(x;f)$, it holds that
\begin{equation}
    \label{eq:appB.2-model_perturbation_optimality_start}
    V_N(x;f) \leq \sum^{N-1}_{k=0}\ell_{k|t}(f^\ast) + \elf(\eta_{N|t}(f^\ast)).
\end{equation}
Consequently, the difference $V_N(x;f)- V_N(x;\tsysn)$ is upper bounded as
\begin{multline}
    \label{eq:appB.2-model_perturbation_expansion}
    V_N(x;f)- V_N(x;\tsysn) \overset{\eqref{eq:appB.2-model_perturbation_optimality_start}}{\leq} \sum^{N-1}_{k=0}\big(\nQ{\eta_{k|t}(f^\ast)} - \nQ{\xi^\star_{k|t}(f^\ast)}\big) \\ 
    + \nP{\eta_{N|t}(f^\ast)} - \nP{\xi^\star_{N|t}(f^\ast)} \\[-0.2cm]
    \overset{\eqref{eq:appB.2-pre_quadratic_lemma}}{\leq}
    \eigmax{Q}\sum^{N-1}_{k=1}\|\eta_{k|t}(f^\ast) - \xi^\star_{k|t}(f^\ast)\|^2 + \eigmax{P}\|\eta_{N|t}(f^\ast) - \xi^\star_{N|t}(f^\ast)\|^2 \\[-0.2cm] + 2\eigmax{Q}\sum^{N-1}_{k=1}\|\xi^\star_{k|t}(f^\ast)\|\|\eta_{k|t}(f^\ast) - \xi^\star_{k|t}(f^\ast)\| \\ + 2\eigmax{P}\|\xi^\star_{N|t}(f^\ast)\|\|\eta_{N|t}(f^\ast) - \xi^\star_{N|t}(f^\ast)\|.
\end{multline}
Next, the terms in~\eqref{eq:appB.2-model_perturbation_expansion} are further relaxed in two separate parts. First, the quadratic terms involving $\|\eta_{k|t}(f^\ast) - \xi^\star_{k|t}(f^\ast)\|^2$ are bounded via the Cauchy-Schwarz inequality as
\begin{multline}
    \label{eq:appB.2-model_perturbation_expansion_derivation_quad}
    \eigmax{Q}\sum^{N-1}_{k=1}\|\eta_{k|t}(f^\ast) - \xi^\star_{k|t}(f^\ast)\|^2 + \eigmax{P}\|\eta_{N|t}(f^\ast) - \xi^\star_{N|t}(f^\ast)\|^2 \\
    \overset{\eqref{eq:appB.2-scaled_state_perturbation_additional}}{\leq}[L_{\Df}(\mmf)]^2\left[\eigmax{Q}\sum^{N-1}_{k=1}\zeta_k\sum^{k-1}_{i=0}\ell^\ast_{i|t}(f^\ast) + \eigmax{P}\zeta_N\sum^{N-1}_{i=0}\ell^\ast_{i|t}(f^\ast)\right] \\
    \leq \underbrace{\left(\eigmax{Q}\sum^{N-1}_{k=1}\zeta_k + \eigmax{P}\zeta_N\right)}_{\coloneqq \beta^{(2)}_{N}}[L_{\Df}(\mmf)]^2\sum^{N-1}_{i=0}\ell^\ast_{i|t}(f^\ast)\\[-0.2cm]
    \leq \beta^{(2)}_{N}[L_{\Df}(\mmf)]^2V_N(x;\tsysn),
\end{multline}
where $\zeta_k$ is given in~\eqref{eq:appB.2-cost_scaled_state_perturbation_2} in Lemma~\ref{lm:state_mismatch_scaled}. On the other hand, the cross terms involving $\|\eta_{k|t}(f^\ast) - \xi^\star_{k|t}(f^\ast)\|$ are bounded as
\begin{align}
    \label{eq:appB.2-model_perturbation_expansion_derivation_linear}
    & \eigmax{Q}\sum^{N-1}_{k=1}\|\xi^\star_{k|t}(f^\ast)\|\|\eta_{k|t}(f^\ast) - \xi^\star_{k|t}(f^\ast)\| \notag \\[-0.2cm] 
    & \hspace{9em} + \eigmax{P}\|\xi^\star_{N|t}(f^\ast)\|\|\eta_{N|t}(f^\ast) - \xi^\star_{N|t}(f^\ast)\| \notag\\
    \leq &\Bigg\{\left(\sum^{N-1}_{k=0}\nQ{\xi^\star_{k|t}(f^\ast)}+\nP{\xi^\star_{N|t}(f^\ast)}\right) \notag \\
    & \hspace{-1.4em}\Bigg(\frac{\eigmax{Q}^2}{\eigmin{Q}}\hspace{-0.3em}\sum^{N-1}_{k=1}\|\eta_{k|t}(f^\ast) \sminus \xi^\star_{k|t}(f^\ast)\|^2 \splus \frac{\eigmax{P}^2}{\eigmin{P}}\|\eta_{N|t}(f^\ast) \sminus \xi^\star_{N|t}(f^\ast)\|^2\Bigg)\hspace{-0.3em}\Bigg\}^{\hspace{-0.3em}\frac{1}{2}}\notag \\
    \overset{\eqref{eq:appB.2-scaled_state_perturbation_additional}}{\leq} &\Bigg\{\left(\sum^{N-1}_{k=0}\nQ{\xi^\star_{k|t}(f^\ast)}+\nP{\xi^\star_{N|t}(f^\ast)}\right) \notag \\
    & \hspace{-1.3em}[L_{\Df}(\mmf)]^2\left[\frac{\eigmax{Q}^2}{\eigmin{Q}}\sum^{N-1}_{k=1}\zeta_k\sum^{k-1}_{i=0}\ell^\ast_{i|t}(f^\ast) + \frac{\eigmax{P}^2}{\eigmin{P}}\zeta_N\sum^{N-1}_{i=0}\ell^\ast_{i|t}(f^\ast)\right]\hspace{-0.3em}\Bigg\}^{\hspace{-0.3em}\frac{1}{2}}\notag \\
    \leq & \Bigg\{\left(\sum^{N-1}_{k=0}\nQ{\xi^\star_{k|t}(f^\ast)}+\nP{\xi^\star_{N|t}(f^\ast)}\right) \notag \\
    & \hspace{3em}\left(\frac{\eigmax{Q}^2}{\eigmin{Q}}\sum^{N-1}_{k=1}\zeta_k + \frac{\eigmax{P}^2}{\eigmin{P}}\zeta_N\right)[L_{\Df}(\mmf)]^2\sum^{N-1}_{i=0}\ell^\ast_{i|t}(f^\ast)\Bigg\}^{\hspace{-0.3em}\frac{1}{2}}\notag \\
    \leq & \left(\frac{\eigmax{Q}^2}{\eigmin{Q}}\sum^{N-1}_{k=1}\zeta_k + \frac{\eigmax{P}^2}{\eigmin{P}}\zeta_N\right)^{\frac{1}{2}}L_{\Df}(\mmf)V_N(x;\tsysn).
\end{align}
Substituting~\eqref{eq:appB.2-model_perturbation_expansion_derivation_quad} and~\eqref{eq:appB.2-model_perturbation_expansion_derivation_linear} into~\eqref{eq:appB.2-model_perturbation_expansion} yields
\begin{multline}
    \label{eq:appB.2-model_perturbation_expansion_derivation_final}
    V_N(x;f)- V_N(x;\tsysn) \leq \beta^{(2)}_{N}[L_{\Df}(\mmf)]^2V_N(x;\tsysn) \\
    + \underbrace{2\left(\frac{\eigmax{Q}^2}{\eigmin{Q}}\sum^{N-1}_{k=1}\zeta_k + \frac{\eigmax{P}^2}{\eigmin{P}}\zeta_N\right)^{\frac{1}{2}}}_{\coloneqq \beta^{(1)}_{N}}L_{\Df}(\mmf)V_N(x;\tsysn) \\
    \leq \left\{\beta^{(1)}_{N}L_{\Df}(\mmf) + \beta^{(2)}_{N}[L_{\Df}(\mmf)]^2\right\}V_N(x;\tsysn),
\end{multline}
which establishes the perturbation bound~\eqref{eq:model_mismatch_original}.
\end{proof}

%% file: 5-appendices/2.2.3-pp3.tex
\begin{proof}[Proof of Proposition~\ref{prop:k_sum_mismatch}]
    The model-driven perturbation on the $k$-sum stage cost $\left|\sum^{k-1}_{i=0}\ell^\ast_{i|t}(f) - \sum^{k-1}_{i=0}\ell_{i|t}(f)\right|$ can be proceeded as
\begin{multline}
    \label{eq:appB.2-k_sum_expansion}
    \left|\sum^{k-1}_{i=0}\ell^\ast_{i|t}(f) - \sum^{k-1}_{i=0}\ell_{i|t}(f)\right| \leq \sum^{k-1}_{i=0}\big|\nQ{\eta^\ast_{i|t}(f)} - \nQ{\xi^\star_{i|t}(f)}\big| \\
    \hspace{-13em}\overset{\eqref{eq:appB.2-pre_quadratic_lemma}}{\leq} \eigmax{Q}\sum^{k-1}_{i=1}\|\eta^\ast_{i|t}(f) \sminus \xi^\star_{i|t}(f)\|^2 \\[-0.3cm]
    + 2\eigmax{Q}\sum^{k-1}_{i=1} \|\xi^\star_{i|t}(f)\|\|\eta^\ast_{i|t}(f) - \xi^\star_{i|t}(f)\|.
\end{multline}
First, the terms involving $\|\eta^\ast_{i|t}(f) - \xi^\star_{i|t}(f)\|^2$ are bounded as
\begin{multline}
    \label{eq:appB.2-k_sum_expansion_derivation_quad}
    \sum^{k-1}_{i=1}\|\eta^\ast_{i|t}(f) - \xi^\star_{i|t}(f)\|^2 \overset{\eqref{eq:appB.2-cost_scaled_state_perturbation_1}}{\leq} \eigmax{Q}[L_{\Df}(\mmf)]^2\sum^{k-1}_{i=1}\bar{\zeta}_i\sum^{i-1}_{j=0}\ell_{j|t}(f) \\
    \leq \left(\sum^{k-1}_{i=1}\bar{\zeta}_i\right)[L_{\Df}(\mmf)]^2 \sum^{k-1}_{i=0}\ell_{i|t}(f).
\end{multline}
Next, the cross terms involving $\|\eta^\ast_{i|t}(f) - \xi^\star_{i|t}(f)\|$ are bounded via the Cauchy-Schwarz inequality as
\begin{align}
    \label{eq:appB.2-k_sum_expansion_derivation_linear}
    & \sum^{k-1}_{i=1}\|\xi^\star_{i|t}(f)\|\|\eta^\ast_{i|t}(f) - \xi^\star_{i|t}(f)\| \notag \\
    \leq & \;\frac{1}{(\eigmin{Q})^{\frac{1}{2}}} \Bigg\{\left(\sum^{k-1}_{i=1}\nQ{\xi^\star_{i|t}(f)}\right)\left(\sum^{k-1}_{i=1}\|\eta^\ast_{i|t}(f) - \xi^\star_{i|t}(f)\|^2\right)\Bigg\}^{\frac{1}{2}} \notag\\
    \overset{\eqref{eq:appB.2-cost_scaled_state_perturbation_1}}{\leq} & \;\frac{1}{(\eigmin{Q})^{\frac{1}{2}}} \Bigg\{\left(\sum^{k-1}_{i=1}\nQ{\xi^\star_{i|t}(f)}\right)[L_{\Df}(\mmf)]^2\sum^{k-1}_{i=1}\bar{\zeta}_i\sum^{i-1}_{j=0}\ell_{j|t}(f)\Bigg\}^{\frac{1}{2}} \notag\\
    \leq & \;\frac{1}{(\eigmin{Q})^{\frac{1}{2}}}\left(\sum^{k-1}_{i=1}\bar{\zeta}_i\right)^{\frac{1}{2}}L_{\Df}(\mmf)\sum^{k-1}_{i=0}\ell_{i|t}(f).
\end{align}
Substituting~\eqref{eq:appB.2-k_sum_expansion_derivation_quad} and~\eqref{eq:appB.2-k_sum_expansion_derivation_linear} into~\eqref{eq:appB.2-k_sum_expansion} yields
\begin{multline}
    \label{eq:appB.2-k_sum_expansion_derivation_final}
    \hspace{-1.2em}\left|\sum^{k-1}_{i=0}\ell^\ast_{i|t}(f) - \sum^{k-1}_{i=0}\ell_{i|t}(f)\right| \leq \underbrace{\eigmax{Q}\left(\sum^{k-1}_{i=1}\bar{\zeta}_i\right)}_{\coloneqq \theta^{(2)}_k}[L_{\Df}(\mmf)]^2 \sum^{k-1}_{i=0}\ell_{i|t}(f) \\[-0.2cm]
    \hspace{-9em} + \underbrace{\frac{2\eigmax{Q}}{(\eigmin{Q})^{\frac{1}{2}}}\left(\sum^{k-1}_{i=1}\bar{\zeta}_i\right)^{\frac{1}{2}}}_{\coloneqq \theta^{(1)}_k}L_{\Df}(\mmf)\sum^{k-1}_{i=0}\ell_{i|t}(f) \\[-0.2cm]
    \leq \left\{\theta^{(1)}_kL_{\Df}(\mmf) + \theta^{(2)}_k[L_{\Df}(\mmf)]^2\right\}\sum^{k-1}_{i=0}\ell_{i|t}(f),
\end{multline}
thereby establishing the perturbation bound~\eqref{eq:model_mismatch_k_sum}.
\end{proof}

%% file: 5-appendices/2.3-final_theorem.tex
\subsection{Proof of Theorem~\ref{thm:stability_performance_acempc}}
\label{appendix:B-3-final_proof}
\begin{proof}
    The proof combines the nominal dissipation inequality~\eqref{eq:appA-dissipation_final_general} with the perturbation bounds~\eqref{eq:state_perturbation_original}--\eqref{eq:model_mismatch_k_sum}. Specifically, the analysis centers on evaluating the dissipation term $V_N(x_{t_{j+1}};f) - V_N(x_{t_j};f)$, which represents the forward difference of the CE-MPC value function at consecutive active-feedback time steps:
\begin{align}
    \label{eq:appB.3-final_dissipation_derivation}
    & V_N(x_{t_{j+1}};f) - V_N(x_{t_{j}};f) = V_N(\eta_{m_j|t_j}(f^\ast);f) -  V_N(x_{t_{j}};f) \notag \\
    & \hspace{1em} = \underbrace{V_N(\eta_{m_j|t_j}(f^\ast);f) - V_N(\xi^\star_{m_j|t_j}(f^\ast);f)}_{\text{state-driven perturbation}} \notag \\ 
    & \hspace{10em} + \underbrace{V_N(\xi^\star_{m_j|t_j}(f^\ast);f) - V_N(x_{t_{j}};f)}_{\text{nominal dissipation}} \notag \\
    & \overset{\eqref{eq:state_perturbation_original};\;\eqref{eq:appA-dissipation_final}}{\leq} \alpha_{N,m_j}(\mmf)\sum^{k-1}_{i=0}\ell_{i|t}(f) -[1 - \delta_{N,m_j}] \sum^{m_j-1}_{k=0}\ell_{k|t_j}(f) \notag \\[-0.2cm]
    & \hspace{1.2em} \leq -[1 - (\delta_{N,m_j} + \alpha_{N,m_j}(\mmf))]\sum^{m_j-1}_{k=0}\ell_{k|t_j}(f) \notag \\
    & \hspace{1em} \overset{\eqref{eq:model_mismatch_k_sum}}{\leq} -\frac{[1 - (\delta_{N,m_j} + \alpha_{N,m_j}(\mmf))]}{1 + \theta_{m_j}(\mmf)}\sum^{m_j-1}_{k=0}\ell^\ast_{k|t_j}(f) \notag \\
    & \hspace{1.2em} \leq -\frac{[1 - (\delta_{N,M} + \alpha_{N,M}(\mmf))]}{1 + \theta_{M}(\mmf)}\sum^{m_j-1}_{k=0}\ell^\ast_{k|t_j}(f).
\end{align}
In view of the stability index $\rho_{N,M}(\mmf)$ defined in~\eqref{eq:stabiltity_index_acempc}, inequality~\eqref{eq:appB.3-final_dissipation_derivation} implies that
\begin{equation}
    \label{eq:appB.3-final_dissipation}
    \hspace{-0.4em}V_N(x_{t_{j+1}};f) - V_N(x_{t_{j}};f) \leq -\frac{1 - \rho_{N,M}(\mmf)}{1 + \theta_{M}(\mmf)}\hspace{-0.1em}\sum^{m_j-1}_{k=0}\hspace{-0.3em}\ell^\ast_{k|t_j}(f).\hspace{-0.3em}
\end{equation}
Consistent with~\eqref{eq:appA-telescopic}, evaluating the telescopic sum of~\eqref{eq:appB.3-final_dissipation} yields
\begin{equation}
    \label{eq:appB.3-telescopic_original}
    \hspace{-0.5em}\frac{1 - \rho_{N,M}(\mmf)}{1 + \theta_{M}(\mmf)}\hspace{-0.1em}\sum^{J}_{j=1}\sum^{m_j-1}_{k=0}\hspace{-0.3em}\ell^\ast_{k|t_j}(f) \leq V_N(x;f) -V_N(x_{t_{J+1}};f).\hspace{-0.3em}
\end{equation}
By taking the limit as $J \to \infty$ and noting that $V_N(x_{t_{J+1}};f) \geq 0$, relation~\eqref{eq:appB.3-telescopic_original} leads to
\begin{multline}
    \label{eq:appB.3-telescopic_limit}
    \frac{1 - \rho_{N,M}(\mmf)}{1 + \theta_{M}(\mmf)}\hspace{-0.1em}\sum^{\infty}_{j=1}\sum^{m_j-1}_{k=0}\hspace{-0.3em}\ell^\ast_{k|t_j}(f) \leq V_N(x;f) \\
    \overset{\eqref{eq:cost_controllability};\;\eqref{eq:model_mismatch_original}}{\leq} [1+\min\{\umpc, \beta_N(\mmf)\}]V_N(x;F^\ast_{0,N}).
\end{multline}
Following standard arguments in MPC stability analysis via relaxed dynamic programming~\cite{grune2008infinite, grune2017nonlinear} (cf. the derivation of Theorem~\ref{thm:nominal_stability_performance} in Appendix~\ref{appendix:A-proof-thm2}), asymptotic stability of the origin under the asynchronous CE-MPC controller is ensured for the system subject to model mismatch, provided that $\rho_{N,M}(\mmf) < 1$.

Due to symmetry, and following the proof structure established for Proposition~\ref{prop:model_mismatch} in Appendix~\ref{appendix:B-2-propositions}, a counterpart to~\eqref{eq:model_mismatch_original} is obtained as:
\begin{equation}
    \label{eq:appB.3-counterpart_model_mismatch_perturbation}
    V_N(x;\tsysn) -  V_N(x;f) \leq \beta_N(\mmf)V_N(x;f),
\end{equation}
which holds for all $x \in \cX$ and $t \in \bN$. Consequently, an extended cost controllability bound for $V_{\infty}(x;\tsysc)$ is established as:
\begin{equation}
     \label{eq:appB.3-extended_cost_controllability}
    V_{\infty}(x;\tsysc) \overset{\eqref{eq:cost_controllability};\;\eqref{eq:appB.3-counterpart_model_mismatch_perturbation}}{\leq} \big[1 + \underbrace{\umpc+\beta_{\infty}(\mmf)+\umpc\beta_{\infty}(\mmf)}_{\coloneqq \bar{\gamma}^\ast(\mmf)}\big]\ell^\star(x).
\end{equation}
Then, analogously to~\cite[Proof of Theorem~5, Part~II]{kohler2023stability}, it holds that
\begin{equation}
    \label{eq:appB.3-mpcvalue_to_infty}
    V_N(x;F^\ast_{0,N}) \leq 1 + \frac{\eigmax{P}}{\eigmin{Q}}\left(\frac{\bar{\gamma}^\ast(\mmf)}{1+\bar{\gamma}^\ast(\mmf)}\right)^N V_{\infty}(x;\tsysc).
\end{equation}
Combining~\eqref{eq:appB.3-telescopic_limit} and~\eqref{eq:appB.3-mpcvalue_to_infty} yields
\begin{multline}
    \sum^{\infty}_{j=1}\sum^{m_j-1}_{k=0}\hspace{-0.3em}\ell^\ast_{k|t_j}(f) \leq \big[1 - \rho_{N,M}(\mmf)\big]^{-1}\big[1+\min\{\umpc,\beta_{N}(\mmf)\}\big]\\ \big[1+\theta_M(\mmf)\big]\left[1 + \frac{\eigmax{P}}{\eigmin{Q}}\left(\frac{\bar{\gamma}^\ast(\mmf)}{1+\bar{\gamma}^\ast(\mmf)}\right)^N\right]V_{\infty}(x;\tsysc),
\end{multline}
thereby establishing the performance bound~\eqref{eq:performance_bound_acempc}.
\end{proof}

%% file: main.bbl
\begin{thebibliography}{10}
\providecommand{\url}[1]{#1}
\csname url@samestyle\endcsname
\providecommand{\newblock}{\relax}
\providecommand{\bibinfo}[2]{#2}
\providecommand{\BIBentrySTDinterwordspacing}{\spaceskip=0pt\relax}
\providecommand{\BIBentryALTinterwordstretchfactor}{4}
\providecommand{\BIBentryALTinterwordspacing}{\spaceskip=\fontdimen2\font plus
\BIBentryALTinterwordstretchfactor\fontdimen3\font minus
  \fontdimen4\font\relax}
\providecommand{\BIBforeignlanguage}[2]{{%
\expandafter\ifx\csname l@#1\endcsname\relax
\typeout{** WARNING: IEEEtran.bst: No hyphenation pattern has been}%
\typeout{** loaded for the language `#1'. Using the pattern for}%
\typeout{** the default language instead.}%
\else
\language=\csname l@#1\endcsname
\fi
#2}}
\providecommand{\BIBdecl}{\relax}
\BIBdecl

\bibitem{worthmann2015model}
K.~Worthmann, M.~W. Mehrez, M.~Zanon, G.~K.~I. Mann, R.~G. Gosine, and
  M.~Diehl, ``Model predictive control of nonholonomic mobile robots without
  stabilizing constraints and costs,'' \emph{IEEE Transactions on Control
  Systems Technology}, vol.~24, no.~4, pp. 1394--1406, 2015.

\bibitem{wang2020event}
B.~Wang, J.~Huang, C.~Wen, J.~Rodriguez, C.~Garcia, H.~B. Gooi, and Z.~Zeng,
  ``Event-triggered model predictive control for power converters,'' \emph{IEEE
  Transactions on Industrial Electronics}, vol.~68, no.~1, pp. 715--720, 2020.

\bibitem{wu2020distributed}
N.~Wu, D.~Li, Y.~Xi, and B.~{D}e Schutter, ``Distributed event-triggered model
  predictive control for urban traffic lights,'' \emph{IEEE Transactions on
  Intelligent Transportation Systems}, vol.~22, no.~8, pp. 4975--4985, 2020.

\bibitem{li2014event}
H.~Li and Y.~Shi, ``Event-triggered robust model predictive control of
  continuous-time nonlinear systems,'' \emph{Automatica}, vol.~50, no.~5, pp.
  1507--1513, 2014.

\bibitem{brunner2017robust}
F.~D. Brunner, W.~P. M.~H. Heemels, and F.~Allg{\"o}wer, ``Robust
  event-triggered {MPC} with guaranteed asymptotic bound and average sampling
  rate,'' \emph{IEEE Transactions on Automatic Control}, vol.~62, no.~11, pp.
  5694--5709, 2017.

\bibitem{sun2019robust}
Z.~Sun, L.~Dai, K.~Liu, D.~V. Dimarogonas, and Y.~Xia, ``Robust self-triggered
  {MPC} with adaptive prediction horizon for perturbed nonlinear systems,''
  \emph{IEEE Transactions on Automatic Control}, vol.~64, no.~11, pp.
  4780--4787, 2019.

\bibitem{grune2010analysis}
L.~Gr{\"u}ne, J.~Pannek, M.~Seehafer, and K.~Worthmann, ``Analysis of
  unconstrained nonlinear {MPC} schemes with time varying control horizon,''
  \emph{SIAM Journal on Control and Optimization}, vol.~48, no.~8, pp.
  4938--4962, 2010.

\bibitem{eqtami2010event}
A.~Eqtami, D.~V. Dimarogonas, and K.~J. Kyriakopoulos, ``Event-triggered
  control for discrete-time systems,'' in \emph{Proceedings of the 2010
  American Control Conference}, 2010, pp. 4719--4724.

\bibitem{heemels2012introduction}
W.~P. M.~H. Heemels, K.~H. Johansson, and P.~Tabuada, ``An introduction to
  event-triggered and self-triggered control,'' in \emph{2012 51st IEEE
  Conference on Decision and Control (CDC)}, 2012, pp. 3270--3285.

\bibitem{rawlings2017model}
J.~Rawlings, D.~Mayne, and M.~Diehl, \emph{Model {P}redictive {C}ontrol:
  {T}heory, {C}omputation, and {D}esign}.\hskip 1em plus 0.5em minus
  0.4em\relax Nob Hill Publishing, 2020.

\bibitem{mesbah2022fusion}
A.~Mesbah, K.~P. Wabersich, A.~P. Schoellig, M.~N. Zeilinger, S.~Lucia, T.~A.
  Badgwell, and J.~A. Paulson, ``Fusion of machine learning and {MPC} under
  uncertainty: What advances are on the horizon?'' in \emph{2022 American
  Control Conference (ACC)}, 2022, pp. 342--357.

\bibitem{kohler2020computationally}
J.~K{\"o}hler, R.~Soloperto, M.~A. M{\"u}ller, and F.~Allg{\"o}wer, ``A
  computationally efficient robust model predictive control framework for
  uncertain nonlinear systems,'' \emph{IEEE Transactions on Automatic Control},
  vol.~66, no.~2, pp. 794--801, 2020.

\bibitem{kohler2021robust}
J.~K{\"o}hler, P.~K{\"o}tting, R.~Soloperto, F.~Allg{\"o}wer, and
  M.~M{\"u}ller, ``A robust adaptive model predictive control framework for
  nonlinear uncertain systems,'' \emph{International Journal of Robust and
  Nonlinear Control}, vol.~31, no.~18, pp. 8725--8749, 2021.

\bibitem{wabersich2022cautious}
K.~Wabersich and M.~Zeilinger, ``Cautious {B}ayesian {MPC}: {R}egret analysis
  and bounds on the number of unsafe learning episodes,'' \emph{IEEE
  Transactions on Automatic Control}, vol.~68, no.~8, pp. 4896--4903, 2022.

\bibitem{liu2025regret}
C.~Liu, S.~Shi, and B.~{D}e Schutter, ``On the regret of model predictive
  control with imperfect inputs,'' \emph{IEEE Control Systems Letters}, vol.~9,
  pp. 601--606, 2025.

\bibitem{degner2026adaptive}
M.~Degner, R.~Soloperto, M.~N. Zeilinger, J.~Lygeros, and J.~K{\"o}hler,
  ``Adaptive economic model predictive control: {P}erformance guarantees for
  nonlinear systems,'' \emph{IEEE Transactions on Automatic Control}, vol.~71,
  no.~7, pp. 4355--4370, 2026.

\bibitem{soloperto2019dual}
R.~Soloperto, J.~K{\"o}hler, M.~A. M{\"u}ller, and F.~Allg{\"o}wer, ``Dual
  adaptive {MPC} for output tracking of linear systems,'' in \emph{2019 IEEE
  58th Conference on Decision and Control (CDC)}, 2019, pp. 1377--1382.

\bibitem{liu2024stability}
C.~Liu, S.~Shi, and B.~{D}e Schutter, ``Stability and performance analysis of
  model predictive control of uncertain linear systems,'' in \emph{2024 IEEE
  63rd Conference on Decision and Control (CDC)}, 2024, pp. 7356--7362.

\bibitem{shi2025suboptimality}
S.~Shi, A.~Tsiamis, and B.~D. Schutter, ``Suboptimality analysis of receding
  horizon quadratic control with unknown linear systems and its applications in
  learning-based control,'' \emph{IEEE Transactions on Automatic Control},
  vol.~71, no.~3, pp. 1422--1437, 2025.

\bibitem{schwenkel2020robust}
L.~Schwenkel, J.~K{\"o}hler, M.~M{\"u}ller, and F.~Allg{\"o}wer, ``Robust
  economic model predictive control without terminal conditions,''
  \emph{IFAC-PapersOnLine}, vol.~53, no.~2, pp. 7097--7104, 2020.

\bibitem{liu2026certainty}
C.~Liu, S.~Shi, and B.~{D}e Schutter, ``Certainty-equivalence model predictive
  control: {S}tability, performance, and beyond,'' \emph{IEEE Transactions on
  Automatic Control}, vol.~71, no.~7, pp. 4649--4664, 2026.

\bibitem{sun2021dynamic}
Z.~Sun, C.~Li, J.~Zhang, and Y.~Xia, ``Dynamic event-triggered {MPC} with
  shrinking prediction horizon and without terminal constraint,'' \emph{IEEE
  Transactions on Cybernetics}, vol.~52, no.~11, pp. 12\,140--12\,149, 2021.

\bibitem{lincoln2006relaxing}
B.~Lincoln and A.~Rantzer, ``Relaxing dynamic programming,'' \emph{IEEE
  Transactions on Automatic Control}, vol.~51, no.~8, pp. 1249--1260, 2006.

\bibitem{lu2015asynchronous}
L.~Lu, ``Asynchronous separable self-triggered model predictive control based
  on relaxed dynamic programming,'' \emph{IFAC-PapersOnLine}, vol.~48, no.~8,
  pp. 948--953, 2015.

\bibitem{lu2022self}
L.~Lu, D.~Limon, and I.~Kolmanovsky, ``Self-triggered {MPC} with performance
  guarantee for tracking piecewise constant reference signals,''
  \emph{Automatica}, vol. 142, p. 110364, 2022.

\bibitem{rawlings2012fundamentals}
J.~B. Rawlings, D.~Angeli, and C.~N. Bates, ``Fundamentals of economic model
  predictive control,'' in \emph{2012 IEEE 51st {C}onference on {D}ecision and
  {C}ontrol (CDC)}.\hskip 1em plus 0.5em minus 0.4em\relax IEEE, 2012, pp.
  3851--3861.

\bibitem{kohler2023stability}
J.~K{\"o}hler, M.~Zeilinger, and L.~Gr{\"u}ne, ``Stability and performance
  analysis of {NMPC}: {D}etectable stage costs and general terminal costs,''
  \emph{IEEE Transactions on Automatic Control}, vol.~68, no.~10, pp.
  6114--6129, 2023.

\bibitem{pannocchia2011conditions}
G.~Pannocchia, J.~B. Rawlings, and S.~J. Wright, ``Conditions under which
  suboptimal nonlinear {MPC} is inherently robust,'' \emph{Systems \& Control
  Letters}, vol.~60, no.~9, pp. 747--755, 2011.

\bibitem{grune2017nonlinear}
L.~Gr{\"u}ne and J.~Pannek, \emph{Nonlinear {M}odel {P}redictive
  {C}ontrol}.\hskip 1em plus 0.5em minus 0.4em\relax Springer, 2017.

\bibitem{schimperna2025data}
I.~Schimperna, K.~Worthmann, M.~Schaller, L.~Bold, and L.~Magni, ``Data-driven
  model predictive control: {A}symptotic stability despite approximation errors
  exemplified in the koopman framework,'' \emph{arXiv preprint
  arXiv:2505.05951}, 2025.

\bibitem{kellett2014compendium}
C.~Kellett, ``A compendium of comparison function results,'' \emph{Mathematics
  of Control, Signals, and Systems}, vol.~26, pp. 339--374, 2014.

\bibitem{wang2018safe}
L.~Wang, E.~A. Theodorou, and M.~Egerstedt, ``Safe learning of quadrotor
  dynamics using barrier certificates,'' in \emph{2018 IEEE International
  Conference on {R}obotics and {A}utomation {(ICRA)}}, 2018, pp. 2460--2465.

\bibitem{dacs2025robust}
E.~Da{\c{s}} and J.~W. Burdick, ``Robust control barrier functions using
  uncertainty estimation with application to mobile robots,'' \emph{IEEE
  Transactions on Automatic Control}, vol.~70, no.~7, pp. 4766--4773, 2025.

\bibitem{ames2016control}
A.~D. Ames, X.~Xu, J.~W. Grizzle, and P.~Tabuada, ``Control barrier function
  based quadratic programs for safety critical systems,'' \emph{IEEE
  Transactions on Automatic Control}, vol.~62, no.~8, pp. 3861--3876, 2016.

\bibitem{liu2025robust}
C.~Liu, A.~Alan, S.~Shi, and B.~{D}e Schutter, ``Robust adaptive discrete-time
  control barrier certificate,'' \emph{arXiv preprint arXiv:2508.08153}, 2025.

\bibitem{kuntz2024beyond}
S.~J. Kuntz and J.~B. Rawlings, ``Beyond inherent robustness: {S}trong
  stability of {MPC} despite plant-model mismatch,'' \emph{IEEE Transactions on
  Automatic Control}, vol.~71, no.~2, pp. 780--792, 2026.

\bibitem{bemporad2002hybrid}
A.~Bemporad, W.~P. M.~H. Heemels, and B.~{D}e Schutter, ``On hybrid systems and
  closed-loop {MPC} systems,'' \emph{IEEE Transactions on Automatic Control},
  vol.~47, no.~5, pp. 863--869, 2002.

\bibitem{li2025learning}
T.~Li, ``Learning-augmented control: Adaptively confidence learning for
  competitive {MPC},'' \emph{arXiv preprint arXiv:2507.14595}, 2025.

\bibitem{grune2008infinite}
L.~Gr{\"u}ne and A.~Rantzer, ``On the infinite horizon performance of receding
  horizon controllers,'' \emph{IEEE Transactions on Automatic Control},
  vol.~53, no.~9, pp. 2100--2111, 2008.

\bibitem{gurobi}
\BIBentryALTinterwordspacing
{Gurobi Optimization, LLC}, \emph{Gurobi Optimizer Reference Manual}, 2024.
  [Online]. Available: \url{https://www.gurobi.com}
\BIBentrySTDinterwordspacing

\bibitem{andersson2019}
J.~Andersson, J.~Gillis, G.H., J.~Rawlings, and M.~Diehl, ``{CasADi} -- {A}
  software framework for nonlinear optimization and optimal control,''
  \emph{Mathematical Programming Computation}, vol.~11, no.~1, pp. 1--36, 2019.

\end{thebibliography}
